\documentclass[sigconf]{acmart}

\usepackage{url}
\usepackage{kantlipsum, lipsum}
\usepackage{dsfont}
\usepackage{subfigure}
\usepackage{subcaption}
\usepackage{bbm}
\usepackage{wrapfig}
\usepackage{enumitem}
\usepackage{microtype}
\usepackage[dvipsnames, table]{xcolor}
\usepackage{amsmath}
\usepackage{mathtools}
\usepackage{hyperref}
\usepackage[capitalize,noabbrev,nameinlink]{cleveref}
\usepackage{multirow}
\newtheorem{theorem}{Theorem}
\newtheorem{innercustomthm}{Theorem}
\newenvironment{customthm}[1]
  {\renewcommand\theinnercustomthm{#1}\innercustomthm}
  {\endinnercustomthm}
\newtheorem{assumption}{Assumption}
\definecolor{rowgrayA}{RGB}{228,228,228}
\definecolor{rowgrayB}{RGB}{242,242,242}
\definecolor{lightblueRow}{RGB}{227,241,250}
\definecolor{lightyellowRow}{RGB}{252,248,220}
\usepackage{etoolbox}
\definecolor{TrieFigPink}{RGB}{240,140,138}
\definecolor{TrieFigBlue}{RGB}{178,200,225}

\crefname{appendix}{Appendix}{Appendices}
\Crefname{appendix}{Appendix}{Appendices}

\makeatletter
\pretocmd\appendix{%
  \crefalias{section}{appendix}%
}{}{}
\makeatother

\usepackage{booktabs}
\usepackage[export]{adjustbox}
\usepackage{makecell}
\usepackage{algorithm}
\usepackage[noend]{algpseudocode}
\usepackage{xspace}
\usepackage{svg}

\usepackage[most,skins,theorems]{tcolorbox}
\tcbset{
  aibox/.style={
    width=\linewidth,
    top=8pt,
    bottom=4pt,
    colback=blue!6!white,
    colframe=black,
    colbacktitle=black,
    enhanced,
    center,
    attach boxed title to top left={yshift=-0.1in,xshift=0.15in},
    boxed title style={boxrule=0pt,colframe=white,},
  }
}
\newtcolorbox{AIbox}[2][]{aibox,title=#2,#1}

\newtcolorbox{prefbox}[1][]{
  enhanced,
  colback=blue!5,
  colframe=blue!40,
  boxrule=1pt,
  arc=3mm,
  left=8pt,
  right=8pt,
  top=8pt,
  bottom=8pt,
  title=#1,
  colbacktitle=blue!35,
  coltitle=white,
  fonttitle=\bfseries,
  attach boxed title to top left={xshift=0mm,yshift=0mm},
  boxed title style={
    sharp corners,
    rounded corners=southeast,
    colframe=blue!40,
    colback=blue!35,
    boxrule=0pt
  }
}
\AtBeginDocument{%
  }

\copyrightyear{2026}
\acmYear{2026}
\setcopyright{cc}
\setcctype{by}
\acmConference[RecSys '26]{20th ACM Conference on Recommender Systems}{September 27-October 02, 2026}{Minneapolis, MN, USA}
\acmBooktitle{20th ACM Conference on Recommender Systems (RecSys '26), September 27-October 02, 2026, Minneapolis, MN, USA}
\acmDOI{10.1145/3773078.3831786}
\acmISBN{979-8-4007-2284-4/2026/09}

\newcommand{\methodemoji}{\textsc{BONSAI} \includegraphics[width=0.015\textwidth]{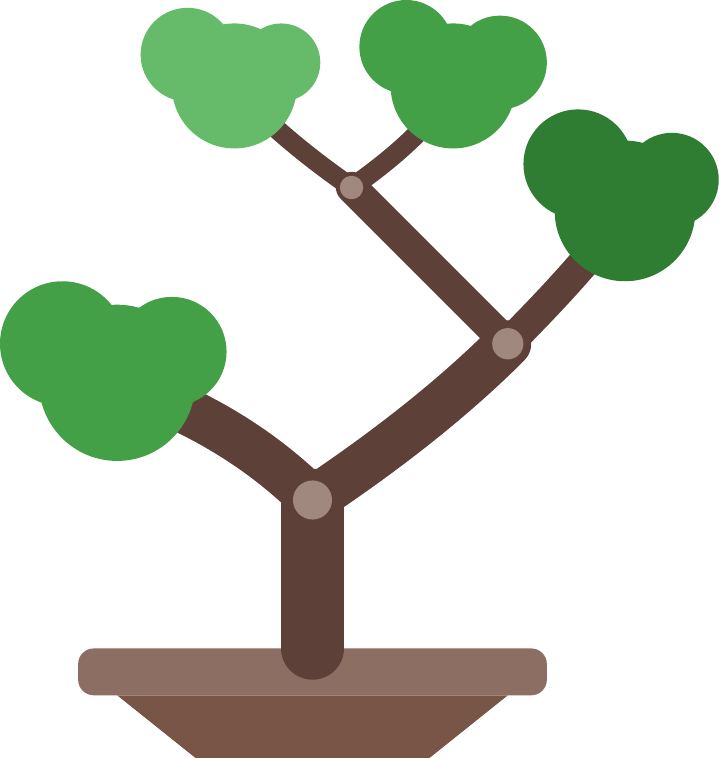}\xspace}

\begin{document}

\title[Optimizing Textual Decoding Tries in LLM-based Generative Recommendation]{Beyond Fixed Depths and Widths: Optimizing Textual Decoding Tries in LLM-based Generative Recommendation}

\author{Jingzhe Liu}
\email{liujin33@msu.edu}
\affiliation{%
  \institution{Michigan State University}
  \city{East Lansing}
  \state{Michigan}
  \country{USA}
}

\author{Hanbing Wang}
\email{wangh137@msu.edu}
\affiliation{%
  \institution{Michigan State University}
  \city{East Lansing}
  \state{Michigan}
  \country{USA}
}

\author{Jiliang Tang}
\email{tangjili@msu.edu}
\affiliation{%
  \institution{Michigan State University}
  \city{East Lansing}
  \state{Michigan}
  \country{USA}
}

\author{Liam Collins}
\email{lcollins2@snapchat.com}
\affiliation{%
  \institution{Snap Inc.}
  \city{Bellevue}
  \state{Washington}
  \country{USA}
  }

\author{Tong Zhao}
\email{tong@snap.com}
\affiliation{%
  \institution{Snap Inc.}
  \city{Bellevue}
  \state{Washington}
  \country{USA}
  }

\author{Neil Shah}
\email{nshah@snap.com}
\affiliation{%
  \institution{Snap Inc.}
  \city{Bellevue}
  \state{Washington}
  \country{USA}
  }

\author{Mingxuan Ju}
\email{mju@snap.com}
\affiliation{%
  \institution{Snap Inc.}
  \city{Bellevue}
  \state{Washington}
  \country{USA}
  }

\renewcommand{\shortauthors}{Jingzhe Liu et al.}

\begin{abstract}
Generative recommendation (GR) is an increasingly popular paradigm in recommender systems, with a prominent line of work using LLMs as autoregressive backbones to predict the next item's term IDs (e.g., titles or keywords).
The success of autoregressive generation hinges on constrained beam search over a \emph{decoding trie} to ensure that generated outputs correspond to valid items. 
However, current research predominantly focuses on generating more comprehensive term IDs to describe items, while largely neglecting the structural design of the decoding trie formed by these terms.
This can lead to a trie that is poorly suited to beam search, which degrades performance.
To address this, we examine the effectiveness of term IDs from the perspective of decoding trie optimization. 
Through empirical and theoretical analyses, we identify two desirable properties for a highly performant trie: \emph{(1) adaptive and variable ID length}, enabling items with varying semantic richness to be represented by IDs of appropriate lengths, and \emph{(2) constrained branching factors}, especially at shallow levels, which drastically improves the success rate of constrained beam search. 
Motivated by these properties, we introduce \textbf{\methodemoji}: \textbf{B}ranching-\textbf{O}ptimized \textbf{N}ode \textbf{S}tructure for \textbf{A}daptive \textbf{I}dentifiers, a novel framework that co-designs textual term IDs and their underlying decoding trie. 
BONSAI extracts recommendation-informative words from item metadata and employs a minimum set cover formulation to recursively build a trie that satisfies the above properties. 
Experiments reveal that BONSAI achieves up to a 21.6\% relative improvement over state-of-the-art baselines. 
Further analyses confirm the crucial role of our proposed properties, and demonstrate their generalizability to be applied to enhance the performance of other term ID methods.
Our codes can be accessed through the link: \textcolor{blue}{\url{https://github.com/Liu-Jingzhe/BONSAI}}.
\end{abstract}


%


\maketitle

\section{Introduction}\label{sec:intro}

Generative recommendation (GR)~\citep{TIGER,HSTU,OneRec} has emerged as a popular paradigm in modern recommender systems. Unlike conventional frameworks that represent each interaction event with a single atomic ID~\citep{huang2020embedding,SASRec,BPR,GRU4Rec,ju2024does,ju2025revisiting,ju2025learning}, GR models an event as a sequence of multiple IDs. In particular, GR based on semantic IDs encodes each event as a sequence of low-cardinality IDs derived from item representations~\citep{TIGER,RQ22,OneRec,ju2025generative,chen2024enhancing}, whereas GR based on term~(textual) IDs represents each event as a sequence of natural-language tokens, such as titles, metadata, or keywords~\citep{LLM4GenRec,P5,how-to-index,liu2025understanding,rec_ID_survey}. At inference time, GR autoregressively generates the ID sequence of the next event, resembling the generation process of large language models (LLMs)~\citep{GPT4,qwen3}. This formulation naturally inherits key advances from LLMs, including architectures~\citep{OneRec-v2,ultra-hstu}, tokenization~\citep{hou2025actionpiece}, and scaling-law behaviors~\citep{liu2025understanding}.

\begin{figure}[t]
\begin{center}
\includegraphics[width=0.5\textwidth]{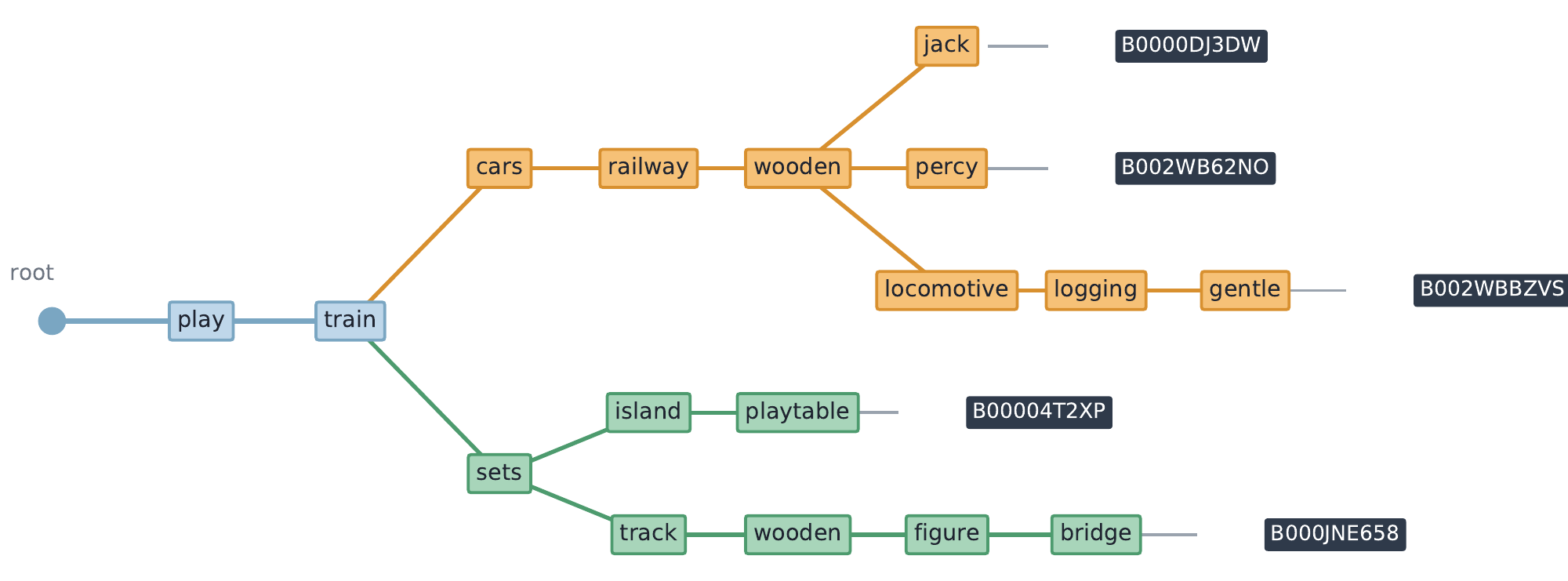}
\end{center}
\vspace{-1.5em}
\caption{A segment of the trie for variable-length term IDs.
}
\label{fig:trie example}
\vspace{-2em}
\end{figure}

For GR based on term IDs, a prominent line of work directly adopts pretrained LLMs as backbones and fine-tunes them on recommendation data to generate item descriptions autoregressively~\citep{openonerec,how-to-index,liu2025understanding,TID,agentictagger,ju2026semantic}. This allows the model to leverage the world knowledge encoded in LLMs, improving both item understanding and user preference modeling~\citep{onerec-think,plum}. However, unlike open-ended text generation, GR requires \emph{grounding}: the model must generate only valid items from a fixed corpus. To ensure this, LLM-based GR relies on \textit{constrained beam search over a decoding trie}~(an example shown in \cref{fig:trie example}) constructed from the tokenized term IDs of all candidate items~\citep{decodingmatters,vectorizing_beam_search,ju2025generative}. Despite the trie directly defining the search space and guiding decoding, its structural properties remain under-explored. Existing studies mainly focus on improving the semantic quality of term IDs~\citep{TID,agentictagger}, while overlooking the fact that semantically meaningful IDs may still induce a trie structure unfavorable to beam search, such as excessive branching near the root, thereby degrading recommendation performance (as shown in \cref{subsec:trie_ablation} and \cref{subsec:optimize_tid} later). Motivated by this gap, we shift the focus from term ID semantics alone to their impact on the induced trie structure, and ask the following question:
\begin{center}
\textbf{\textit{How to co-design term IDs and their decoding trie?}}
\end{center}
\noindent To answer this question, we empirically and theoretically analyze the structure of different decoding tries. 
We identify two desirable properties:

\begin{enumerate}[leftmargin=*,topsep=0pt]

\item \textbf{Adaptive and variable ID length.} 
As illustrated in \cref{fig:word_dis}, the metadata length of items across datasets varies drastically, ranging from fewer than 10 words to more than hundreds of words. Conversely, existing methods~\citep{TIGER,minionerec,TID} typically enforce rigid and uniform token lengths for all item IDs, which inevitably introduces information loss. 
Forcing a uniform structure is both mathematically and empirically suboptimal. 
Furthermore, since an item's ID is equivalent to its corresponding root-to-leaf path in the trie, this property also indicates that \textit{the trie should have variable depths}.

\item \textbf{Constrained branching factors, especially in shallow levels.} 
Constrained beam search's success relies heavily on minimizing the branching factors at early decision nodes. 
During decoding, the generative model must decide which trie paths to explore, as shown in \cref{fig:pre_on_branch}.
If the branching factor (the average number of children) is excessively high at shallow levels, the model's likelihood of selecting the correct path drops significantly.  
\end{enumerate}

\noindent Based on these two properties, we propose a simple yet effective method for constructing search-trie-aware term IDs, dubbed \methodemoji, \textbf{B}ranching-\textbf{O}ptimized \textbf{N}ode \textbf{S}tructure for \textbf{A}daptive \textbf{I}dentifiers, a novel framework that co-designs term IDs and their underlying decoding trie. 
Similar to existing works~\citep{IDGenRec,TID,agentictagger}, we first filter the metadata (i.e., raw textual descriptions) to retain only highly informative terms. 
We then treat each remaining term as a feature, representing each item as a set of such features.
Next, we construct a trie by assigning features to internal nodes; each leaf node corresponds to an item containing all the features along the path from the root. 

The key novelty of BONSAI is that we formulate this construction process as a recursive set cover problem: for each internal node, we select a minimal set of features to split the current node into child nodes that together cover all items associated with that node. 
By explicitly minimizing the cover set, we guarantee that the branching factor at each level remains optimally low. 
Simultaneously, items with richer, more complex feature sets naturally require more splits to be uniquely identified, gracefully pushing them to greater depths in the trie. Our contributions are summarized as:
\begin{itemize}[leftmargin=*,topsep=0pt]
    \item We propose a new perspective to study the term ID construction problem through decoding trie optimization, identifying two favorable structural properties that directly lead to better generation efficiency and accuracy.
    \item We provide comprehensive theoretical guarantees and empirical analyses to validate our proposed properties, demonstrating their generalizability by showing that they also enhance existing term ID methods.
    \item Following these properties, we design BONSAI, a simple yet effective method to construct term IDs and their corresponding trie. 
    Extensive experiments show that our method substantially surpasses state-of-the-art baselines, achieving up to a 21.6\% relative improvement.
\end{itemize}

\section{Related Works}\label{sec:related works}

\noindent\textbf{LLM-based Generative Recommendation.} GR casts recommendation as a language generation task. 
Early frameworks like P5~\citep{P5} and InstructRec~\citep{InstructRec} showed that instruction-tuned language models can be a suitable interface across diverse recommendation objectives.
Motivated by the architectural and scaling properties of pretrained LLMs, recent approaches increasingly trend towards utilizing LLMs as the core recommendation backbone~\citep{GenRec,LC-Rec}, and representing historical user interactions as input prompts and autoregressively decoding textual or semantic identifiers of the next item~\citep{TIGER,minionerec,TID}, unlike classical dot-product scoring between user and item latent  vectors~\citep{covington2016deep,yi2019sampling,yang2020mixed}.
Recent studies further investigate whether LLM reasoning improves preference modeling~\citep{onerec-think,Gream,why-think-hurts}, while others extend beyond single-turn predictions towards LLM-driven recommender agents that support dialogue, tool use, and goal-oriented planning to address more complex real-world recommendation scenarios~\citep{agent4rec,agentcf,iagent,memrec}.


\noindent \textbf{Item Representations in GR.} LLM-based GR methods typically represent items in two ways. One line of work uses semantic IDs (SIDs)~\citep{ju2025generative,openonerec,onerec-think,plum}, which discretize item description embeddings into numerical codes~\citep{RQ22} and expand the LLM vocabulary to include these new tokens. However, SID-based methods can incur information loss from discretization~\citep{liu2025understanding}, require additional alignment to connect the new tokens with the LLM's pretrained knowledge space~\citep{minionerec}, and still face limitations in generalization and transferability~\citep{liu2025bridge,hu2026ids}. Another line of work instead uses textual item representations (or term IDs)~\citep{TALLRec,LlamaRec,LLaRA}, reusing natural language tokens already in the LLM vocabulary. This avoids the overhead of introducing new item tokens and better preserves item semantics, enabling stronger integration of domain-specific recommendation signals with the LLM's open-world knowledge~\citep{TID} and yielding better empirical performance and scaling behavior~\citep{liu2025understanding}.


\noindent\textbf{Constrained Decoding and Search Tries.} Unlike open-ended text generation, GR must produce valid items from a fixed candidate set. Accordingly, LLM-based GR methods commonly use constrained beam search at inference time~\citep{vectorizing_beam_search,ReRe}, where decoding is restricted by a prefix tree (trie) built over the semantic or term IDs of candidate items~\citep{ju2025generative}. Although this trie defines the search space and directly affects decoding behavior (and hence accuracy), its structure has received little attention. Prior work mainly focuses on improving the semantic quality of term IDs, for example, through better metadata or keyword selection~\citep{TID,agentictagger}, while largely treating the induced trie as incidental. We argue this is an important omission: even semantically informative term IDs can be poorly suited to beam search if they induce unfavorable trie structures, such as excessive early branching, that hinder decoding and reduce recommendation quality. Our work addresses this gap by studying term IDs not only as semantic representations, but also as design choices that shape the underlying search process in GR.

\section{Preliminary Studies}\label{sec:preliminary}
In this section, we present empirical findings that motivate the structural design principles of decoding tries for GR. By analyzing both the distribution of item metadata and the behavior of constrained tree search, we identify two key requirements for an effective decoding structure.




\vspace{1.5mm}
\noindent \textbf{Adaptive, Variable-Length Term IDs.} Our first study examines the semantic variability of item representations. As shown in \cref{fig:word_dis}, item metadata in the Amazon Review dataset~\citep{SASRec} varies substantially in length, ranging from fewer than 10 words to hundreds. This variation appears both within a single category and across domains.  This suggests strong variation in the semantic richness across items and their effective representational needs.

\begin{figure}[t]
\begin{center}
\includegraphics[width=0.5\textwidth]{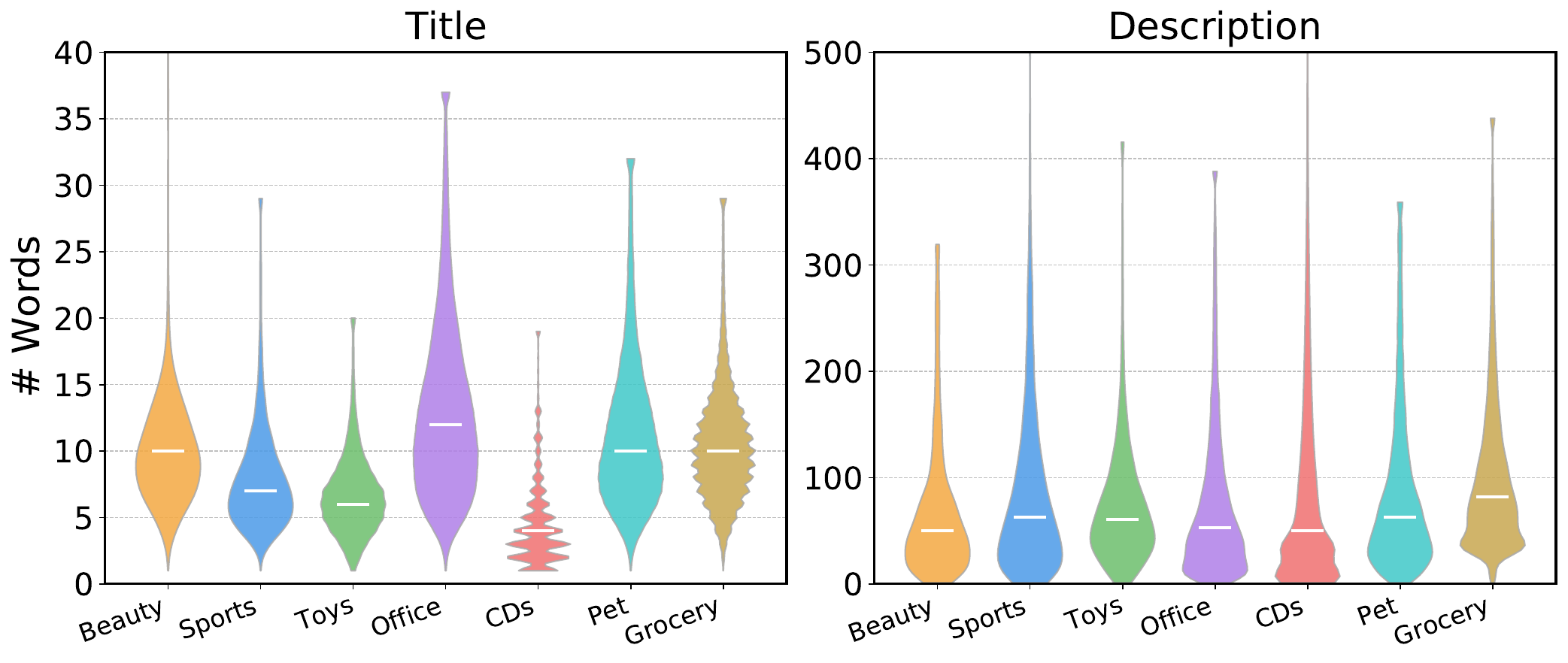}
\end{center}
\vspace{-1.5em}
\caption{The length distribution of item metadata in Amazon Review Datasets~\citep{SASRec}. Semantic richness of different items varies largely both within and across datasets.}
\vspace{-0.7em}
\label{fig:word_dis}
\end{figure}

\begin{figure}[t]
\begin{center}
\vspace{-0.5em}
\includegraphics[width=0.5\textwidth]{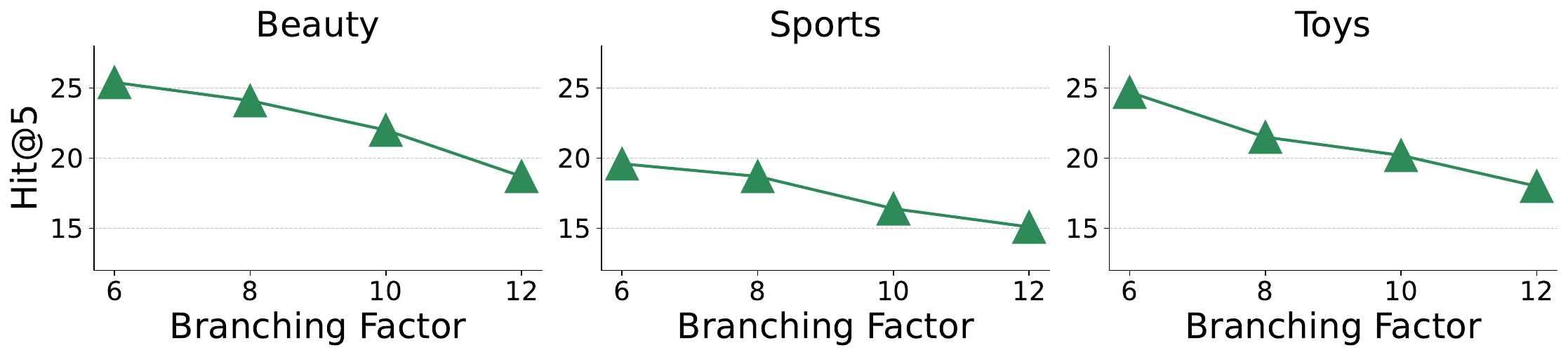}
\end{center}
\vspace{-1.5em}
\caption{Preliminary experiments on predicting the category of the next item under different search tree structures. As the branching factor increases, performance gradually declines.
}
\vspace{-1.5em}
\label{fig:pre_on_branch}
\end{figure}

To reveal the cost of ignoring this variability, we first provide a theoretical analysis~(\cref{thm:variable-depth}) to show that uniform ID length will result in suboptimal performance.
Furthermore, we performed an ablation study later in which we truncated a variable-depth trie built with term IDs to produce shallower, fixed-depth structures~(more detailed settings can be found in \cref{subsec:trie_ablation}). As shown in \cref{fig:abla_depth}, enforcing a uniform depth consistently reduces recommendation performance. Fixed-depth structures compress semantically rich items too aggressively, leading to substantial information loss, while also imposing unnecessary structural constraints on simpler items. 

\noindent \textbf{Takeaway}: \emph{The length of an item's term ID should reflect its semantic richness. To preserve representational fidelity, GR systems should replace uniform, fixed-depth structures with variable-length identifiers that adapt to the underlying metadata.}



\vspace{1.5mm}
\noindent \textbf{Decoding Tries with Constrained Branching Factor.} Our second study isolates the effect of tree width on decoding performance. We begin with a simple task that has a natural hierarchical structure: predicting the exact category of the next item. By merging or splitting subcategories, we construct category trees with different branching factors. As shown in \cref{fig:pre_on_branch}, Hit@5 declines steadily as the branching factor increases from 6 to 12.

We observe the same pattern in the controlled trie experiments in \cref{fig:abla_branch}. When we artificially increase the branching factor using fixed-size cover sets while keeping the depth unchanged, recommendation accuracy decreases consistently across all three datasets. A large branching factor forces the GR model to choose among too many candidates at early decision stages, reducing the probability of following the correct path. Once the search deviates early, the correct item can be excluded from all remaining beams.

\noindent \textbf{Takeaway}: \emph{The success of constrained beam search depends critically on keeping the branching factors small, especially at shallow decision nodes. Large branching factors make early decisions harder and directly reduce search hit rates. Controlling the branching factor is therefore essential for maintaining reliable decoding trajectories and avoiding cascading errors.}


\begin{figure}[t]
\begin{center}
\includegraphics[width=0.46\textwidth]{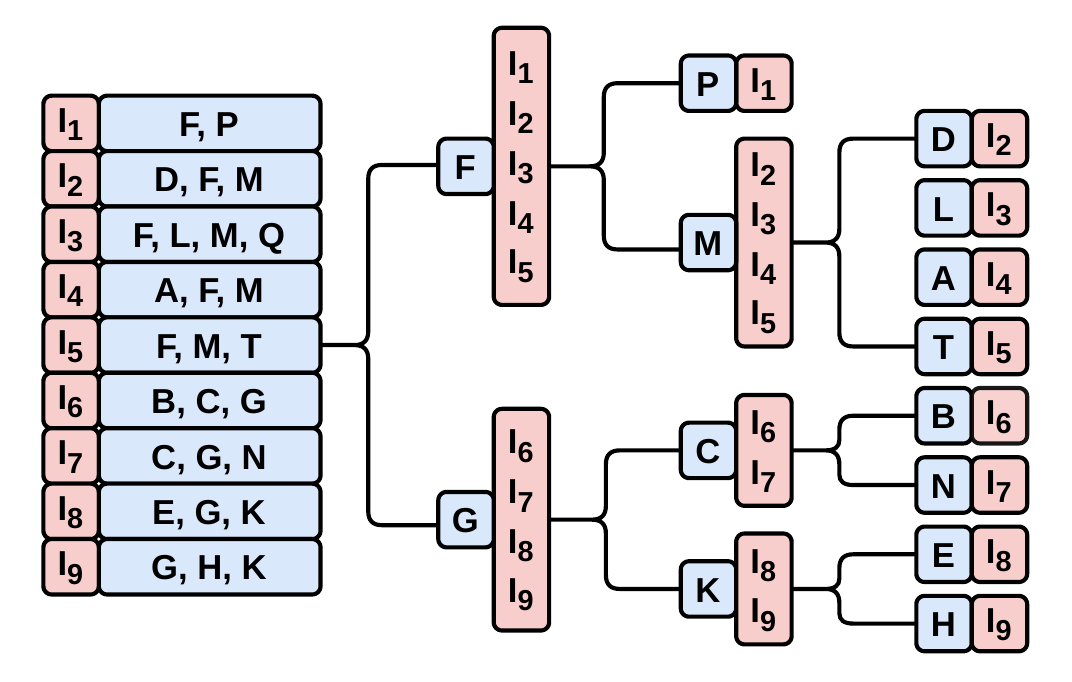}
\end{center}
\vspace{-0.5em}
\caption{Illustration of trie construction process. At each step, we split the {\color{TrieFigPink} item sets} into smaller subsets according to selected {\color{TrieFigBlue}shared features}. Accordingly, each feature serves as the identifier of a node, which contains a subset of items. And each item's ID is equivalent to the sequence of features along the path from the root to its leaf.
}
\label{fig:trie construction}
\vspace{-1.5em}
\end{figure}

\section{Theoretical Analysis}\label{sec:theory}

Here, we provide theoretical analyses to better ground our proposed structural properties for decoding tries in GR. 
We first define the trie construction problem to be solved and notations used in the discussion.

\noindent \textbf{\textit{Problem Setting}}:
Let $I=\{1,\dots,n\}$ be the item set, and let $Y\in I$ denote the correct target items. 
Each item has a set of features, which can be shared or unique.
As shown in \cref{fig:trie construction}, a trie is constructed recursively by partitioning items into subsets with shared features. 
For any nonempty subset $U\subseteq I$, let $\Pi(U)$ denote the set of all feasible feature-induced splits of $U$. 
Each split $\mathcal P \coloneqq \{C_1,\dots,C_k\}\in \Pi(U)$ is a partition of $U$, where $k=|\mathcal P|$ is the number of children. 
Suppose the constrained beam search has beam number $B\ge 1$, and the one-step beam survival factor (the probability that the paths towards correct items are still included by the beam search at a given step) is $\phi_B(\cdot)$. 
For any feasible trie $T$, let $d(i)$ be the depth of leaf $i$, and let $k_t(i)$ be the branching factor of item $i$ at depth $t$. 
The beam search correct rate over $T$ is:
\[
R(T) \coloneqq \sum_{i\in I} p_i \prod_{t=1}^{d(i)} \phi_B\!\bigl(k_t(i)\bigr).
\]
where $p_i$ is the prior probability that item $i$ is the target. 
Our objective is to \textbf{find a recursive partitioning of $I$ that builds a trie $T$ achieving the optimal beam search success rate}, $V(I)\coloneqq \max R(T)$.


\noindent \textbf{\textit{Theoretical Results}}:
Through theoretical analysis, we provide formal guarantees for our adaptive term decoding trie, explicitly validating our two core properties: constrained branching factors and adaptive term ID length.
We begin with the idealized setting in which the prior probability $p_i$ of each item is known. 
In this scenario, we prove that a theoretically optimal solution maximizing the beam search success rate can always be obtained via Bellman dynamic programming~\citep{bertsekas2012dynamic}, as stated in \cref{thm:optimal-tree} below.

\begin{figure*}[t]
\begin{center}
\includegraphics[width=\textwidth]{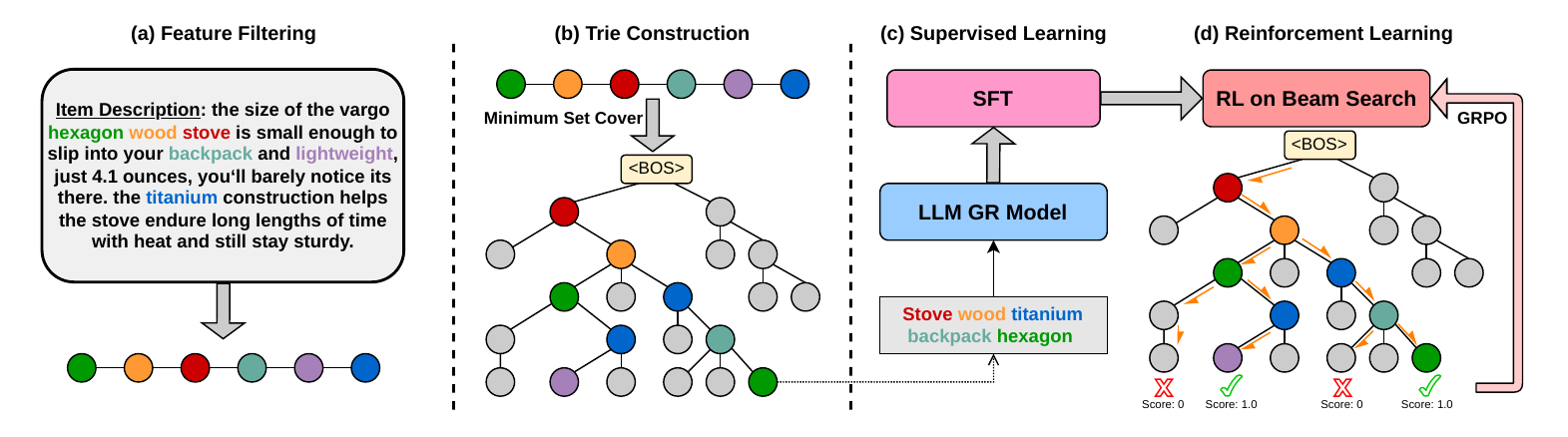}
\end{center}
\vspace{-1.5em}
\caption{Overview of BONSAI and its four main components: (a) feature filtering, which extracts recommendation-relevant terms; (b) trie construction, which builds a term-based trie using a minimum set cover strategy; (c) supervised learning, which trains the LLM backbone to acquire recommendation capability; and (d) reinforcement learning, which optimizes the model under real decoding constraints while aligning items with multiple possible IDs.}
\label{fig:framework}
\vspace{-1em}
\end{figure*}

\begin{theorem}[Optimal tree construction]
\label{thm:optimal-tree}
For any nonempty subset $U\subseteq I$, $V(U)$ denotes the maximum achievable beam search success rate (supposing target items $Y\in U$). 
For any feasible split $\mathcal P\in\Pi(U)$, $\operatorname{Val}(U,\mathcal P)$ denotes the success rate obtained by first applying split $\mathcal P$ at $U$ and then continuing optimally in each resulting child subset. 
For every non-singleton subset $U\subseteq I$,
\[
V(U)
=
\max_{\mathcal P\in \Pi(U)}
\operatorname{Val}(U,\mathcal P).
\]
Choosing at each subset $U$ a maximizing split in $\Pi(U)$ with Bellman dynamic programming yields a globally optimal tree, and the optimal correct rate on the full item set is $V(I)$.
\end{theorem}

However, in real-world recommendation scenarios, the true prior $p_i$ is unavailable. 
To adhere to our design principle of constraining the branching factor at early decision nodes, we formulate the node-splitting process as a minimum set cover problem.
\cref{thm:local-approx} demonstrates that our proposed greedy minimum cover set method provides a strong approximation to the theoretical optimum. 
Specifically, it approximates the optimal solution at each split with a high approximation factor, mathematically ensuring that the fan-out remains optimally low.

\begin{theorem}[Local approximation of greedy set cover]
\label{thm:local-approx}
Under the regularity conditions stated in the appendix, for every non-singleton subset $U\subseteq I$, let
\[
\mathcal P^*(U)\in \arg\max_{\mathcal P\in\Pi(U)} \operatorname{Val}(U,\mathcal P)
\]
be a locally optimal split. The greedy set-cover split $\mathcal P^g(U)$ satisfies
\[
\operatorname{Val}(U,\mathcal P^g(U))
\ge
\frac{1}{\rho \alpha_U}
\operatorname{Val}(U,\mathcal P^*(U)),
\]
where $\rho\ge 1$ is a problem-dependent constant, and $\alpha_U$ indicates how close the greedy algorithm is to an exact solution. 
Thus, the greedy set cover is a provable local approximation when the true prior $\{p_i\}$ is unknown.
\end{theorem}

Furthermore, to theoretically validate our property that an item's term ID length should be inherently adaptive to its semantic richness, we prove in \cref{thm:variable-depth} that both the optimal theoretical construction and our greedy empirical method naturally admit variable leaf depths~(i.e., term ID length). 
Conversely, forcing a fixed-depth trie structure is suboptimal and inevitably degrades the theoretical beam search success rate.

\begin{theorem}[Variable depth is preferable to fixed depth]
\label{thm:variable-depth}
Let $\mathcal T$ be the set of all feasible tries, and let $\mathcal T_{\mathrm{eq}}\subseteq \mathcal T$ be the subset of feasible trees whose leaves have the same depth. Then
\[
\max_{T\in\mathcal T_{\mathrm{eq}}}R(T)\le \max_{T\in\mathcal T}R(T).
\]
Moreover, under the nontrivial-extension condition stated in the appendix, if a feasible tree $T$ has leaves of different depths, then any equal-depth tree obtained by extending the shallower leaves of $T$ satisfies 
\[
R(\widetilde T)<R(T).
\]
Therefore, both the optimal construction and the greedy approximation naturally admit variable leaf depths, while forcing fixed depth will be suboptimal.
\end{theorem}

\noindent
Proofs of the above theorems can be found in \cref{app:proof}.

\section{\underline{B}ranching-\underline{O}ptimized \underline{N}ode \underline{S}tructure for \underline{A}daptive \underline{I}dentifiers (\methodemoji)}\label{sec:method}

In this section, we introduce our methodological contributions. 
As shown in \cref{fig:framework},
our framework consists of four stages: (1) feature filtering: extracting informative keywords from item metadata;  (2) trie construction: building the decoding trie with our proposed properties;
(3) supervised learning: aligning an LLM with a sequential recommendation task;
(4) reinforcement learning: further optimizing the model under real decoding scenarios.

\subsection{Feature Filtering}

To construct a decoding trie that supports adaptive and variable-length term IDs, we first extract terms from raw item metadata to serve as the basic elements in the trie. Specifically, we prompt an LLM to read each item description and remove uninformative words. For example, as shown in \cref{fig:framework}.a, given a dense product description, the LLM identifies salient features such as ``stove,'' ``wood,'' ``titanium,'' ``backpack,'' and ``hexagon,'' while filtering out generic filler terms (a sample prompt is provided in \cref{app:prompt}).

To preserve the natural semantic variability across items, we place no constraint on the number of extracted terms. As a result, each item is represented after filtering as a variable-size set of features. We further restrict the representation to terms that explicitly appear in the original metadata. Unlike prior methods~\citep{TID,agentictagger}, which use LLM annotators to expand or enrich term semantics, our approach relies purely on extraction. This design choice helps isolate the source of BONSAI's gains: they arise primarily from a better-structured decoding trie, rather than from semantically enhanced term IDs. This also points to future opportunities to achieve further gains from enrichment with joint methods.

\subsection{Trie Construction}

After feature filtering, each item is represented as a variable-size set of features, or terms. The next step is to build a decoding trie over these features, where internal nodes correspond to shared features and leaf nodes correspond to items. This trie defines the valid output space of the GR model at inference time and constrains the beam search process, as illustrated in \cref{fig:trie example} and \cref{fig:framework}.b.

BONSAI constructs the trie by recursively partitioning the item set into smaller subsets according to shared features. As shown in \cref{fig:trie construction} and \cref{alg:recursive-set-cover}, we formulate the node-splitting step as a recursive set cover problem. For a given subset of items, the goal is to find the smallest set of features that fully partitions the subset (Line 4 of \cref{alg:recursive-set-cover}). In principle, this step can be solved exactly using a minimum set cover solver such as CP-SAT~\citep{CP-SAT}. In practice, however, the scale of recommendation datasets makes exact optimization prohibitively expensive. We therefore use an efficient greedy approximation (\cref{alg:greedy-cover-split}). At each iteration, the greedy procedure selects the shared feature that covers the largest number of uncovered items (Line 4 of \cref{alg:greedy-cover-split}) and updates the remaining item pool accordingly (Line 9), repeating this process until all items in the subset are covered (Lines 3--9). Once the cover set is obtained, BONSAI adds the corresponding feature groups as child nodes and recursively builds a subtree for each group (Lines 9--11 of \cref{alg:recursive-set-cover}). After the trie is complete, an item's term ID is defined as the sequence of keywords along the path from the root to its leaf node. Importantly, because an item may share different valid combinations of features with different item groups, the construction naturally permits multiple valid root-to-leaf paths for the same item.

\begin{algorithm}[t]
\caption{Trie Construction}
\label{alg:recursive-set-cover}
\begin{algorithmic}[1]
\Require Item set $U \subseteq I$. For each subset $U$, a candidate shared-feature family $\mathcal{F}(U)$; every item also has its own unique feature
\Ensure A variable-depth trie $T(U)$

\If{$|U| = 1$}
    \State Create a leaf for the unique item in $U$
    \State \Return $T(U)$
\EndIf

\State Solve a minimum set cover
\[
\mathcal{P}(U)=\{C_1,\dots,C_k\}
\]
by a given solver~(e.g., Algorithm~\ref{alg:greedy-cover-split})

\If{$\mathcal{P}(U)=\emptyset$}
    \State Create leaves for all items in $U$ using their unique features
    \State \Return $T(U)$
\EndIf

\State Create an internal node associated with the current subset $U$

\For{each child subset $C_j \in \mathcal{P}(U)$}
    \State Recursively construct subtree $T(C_j)$
    \State Attach $T(C_j)$ as a child of the current node
\EndFor

\If{$U \setminus \bigcup_{j=1}^k C_j \neq \emptyset$}
    \State Attach each remaining uncovered item as a leaf using its unique feature
\EndIf

\State \Return $T(U)$
\end{algorithmic}
\end{algorithm}

\begin{algorithm}[t]
\caption{Greedy Minimum Cover Split at a Node}
\label{alg:greedy-cover-split}
\begin{algorithmic}[1]
\Require Current item subset $U \subseteq I$. Candidate shared-feature family $\mathcal{F}(U)=\{F_1,\dots,F_m\}$ restricted to $U$
\Ensure A greedy feature cover $\mathcal{P}^g(U)$ of $U$

\State $\mathcal{U}_{\mathrm{rem}} \gets U$ \Comment{items not yet covered}
\State $\mathcal{P}^g(U) \gets \emptyset$

\While{$\mathcal{U}_{\mathrm{rem}} \neq \emptyset$}
    \State Choose
    $
    F^* \in \arg\max_{F \in \mathcal{F}(U)}
    \left|F \cap \mathcal{U}_{\mathrm{rem}}\right|
    $
    \State $C^* \gets F^* \cap \mathcal{U}_{\mathrm{rem}}$
    \If{$C^* = \emptyset$}
        \State \textbf{break}
    \EndIf
    \State $\mathcal{P}^g(U) \gets \mathcal{P}^g(U) \cup \{C^*\}$
    \State $\mathcal{U}_{\mathrm{rem}} \gets \mathcal{U}_{\mathrm{rem}} \setminus C^*$
\EndWhile

\State \Return $\mathcal{P}^g(U)$
\end{algorithmic}

\end{algorithm}

This construction operationalizes the two structural principles introduced in \cref{sec:preliminary}. First, by computing a minimum cover set at each node, BONSAI keeps the branching factor as low as possible. Limiting fan-out, particularly near the root of the trie, is crucial for decoding: it reduces early ambiguity and increases the probability that constrained beam search remains on the correct path. Second, the recursive partitioning process naturally yields a trie with adaptive depth. Items with simple feature sets are separated quickly, while items with richer semantics generally require more splits to be uniquely identified and therefore appear deeper in the trie. This adaptive structure avoids the information loss caused by fixed-depth trees. In addition, as we show in \cref{subsec:RL ablation}, semantically rich items benefit from having multiple valid representation paths, which further improves retrieval performance.

\subsection{Supervised Learning}
In the initial training stage, we apply supervised fine-tuning (SFT) to adapt the LLM backbone to the sequential recommendation task. Given a user's interaction history $\mathcal{S}_{u}=(i_{1},i_{2},\ldots,i_{t})$, we serialize the historical items using their trie-aware term IDs and train the model to autoregressively generate the term ID of the next item, $i_{t+1}$. Because a single item may correspond to multiple valid root-to-leaf paths in the trie, we use the longest available representation as the supervision target during SFT. This choice preserves more fine-grained semantic information for items with rich feature sets.  Practically, for each training example, we construct a prompt $x$ that encodes the user's interaction history and the recommendation query (a sample prompt is provided in \cref{app:prompt}). Let $y=(y_{1},\ldots,y_{m})$ denote the token sequence of the target item's term ID. The SFT objective is
$
\mathcal{L}_{\mathrm{SFT}}
=
-\sum_{k=1}^{m}\log P_{\theta}(y_{k}\mid x,y_{<k}),
$
where $\theta$ denotes the parameters of the GR model. After this stage, the model learns an initial mapping from user behavior sequences to next-item term IDs.

\begin{table*}[t]
\centering
\caption{Performance of BONSAI compared to various types of methods.  “*” indicates that our method achieves statistically significant improvements (i.e.,
t-test with $p < 0.05$) over the best baseline performance which is \underline{underlined}.}
\vspace{-1em}
\label{tab:main_results}
\small
\setlength{\tabcolsep}{3.5pt}
\renewcommand{\arraystretch}{1.15}
\begin{tabular}{p{1.5cm}|l|cccc|cccc|cccc}
\toprule
\multicolumn{2}{c|}{\multirow{2}{*}{\textbf{Model}}} &
\multicolumn{4}{c|}{\textbf{Beauty}} &
\multicolumn{4}{c|}{\textbf{Sports}} &
\multicolumn{4}{c}{\textbf{Toys}} \\
\cmidrule(lr){3-6} \cmidrule(lr){7-10} \cmidrule(lr){11-14}
\multicolumn{2}{c|}{} &
\textbf{R@5} & \textbf{R@10} & \textbf{N@5} & \textbf{N@10} &
\textbf{R@5} & \textbf{R@10} & \textbf{N@5} & \textbf{N@10} &
\textbf{R@5} & \textbf{R@10} & \textbf{N@5} & \textbf{N@10} \\
\midrule
\multirow{2}{*}{\makecell[l]{\textbf{Traditional}}}
& SASRec       & 0.0402 & 0.0607 & 0.0254 & 0.0320 & 0.0199 & 0.0301 & 0.0106 & 0.0141 & 0.0448 & 0.0626 & 0.0300 & 0.0358 \\
& GRU4Rec      & 0.0395 & 0.0584 & 0.0265 & 0.0326 & 0.0190 & 0.0312 & 0.0122 & 0.0161 & 0.0330 & 0.0490 & 0.0228 & 0.0279 \\
\midrule
\multirow{2}{*}{\makecell[l]{\textbf{Train-from-}\\\textbf{scratch GR}}}
& TIGER        & 0.0405 & 0.0623 & 0.0267 & 0.0337 & 0.0215 & 0.0347 & 0.0137 & 0.0179 & 0.0337 & 0.0547 & 0.0209 & 0.0276 \\
& HSTU         & 0.0424 & 0.0652 & 0.0280 & 0.0353 & 0.0268 & 0.0343 & 0.0173 & 0.0226 & 0.0366 & 0.0566 & 0.0245 & 0.0309 \\
\midrule
\multirow{2}{*}{\makecell[l]{\textbf{LLM-based}\\\textbf{SID GR}}}
& MiniOneRec   & 0.0441 & 0.0643 & 0.0298 & 0.0330 & 0.0252 & 0.0362 & 0.0164 & 0.0192 & 0.0459 & 0.0645 & 0.0314 & 0.0466 \\
& OneRec-Think & 0.0563 & 0.0791 & 0.0398 & 0.0471 & 0.0288 & 0.0412 & 0.0199 & 0.0239 & 0.0579 & 0.0797 & 0.0412 & 0.0482 \\
\midrule
\multirow{2}{*}{\makecell[l]{\textbf{LLM-based}\\\textbf{textual GR}}}
& IDGenRec     & 0.0484 & 0.0693 & 0.0337 & 0.0404 & 0.0270 & 0.0388 & 0.0185 & 0.0223 & 0.0595 & 0.0800 & 0.0432 & 0.0498 \\
& GRLM       & \underline{0.0582} & \underline{0.0804} & \underline{0.0416} & \underline{0.0483} & \underline{0.0326} & \underline{0.0449} & \underline{0.0225} & \underline{0.0278} & \underline{0.0640} & \underline{0.0866} & \underline{0.0440} & \underline{0.0508}\\
\midrule
\multirow{2}{*}{\makecell[l]{\textbf{Ours}}}
& \cellcolor{lightblueRow}\textbf{BONSAI}
& \cellcolor{lightblueRow}\textbf{0.0691*}
& \cellcolor{lightblueRow}\textbf{0.0960*}
& \cellcolor{lightblueRow}\textbf{0.0493*}
& \cellcolor{lightblueRow}\textbf{0.0577*}
& \cellcolor{lightblueRow}\textbf{0.0379*}
& \cellcolor{lightblueRow}\textbf{0.0524*}
& \cellcolor{lightblueRow}\textbf{0.0261*}
& \cellcolor{lightblueRow}\textbf{0.0324*}
& \cellcolor{lightblueRow}\textbf{0.0776*}
& \cellcolor{lightblueRow}\textbf{0.1053*}
& \cellcolor{lightblueRow}\textbf{0.0528*}
& \cellcolor{lightblueRow}\textbf{0.0614*} \\
& \cellcolor{lightyellowRow}\textbf{Impr. (\%)}
& \cellcolor{lightyellowRow}18.7\%
& \cellcolor{lightyellowRow}19.4\%
& \cellcolor{lightyellowRow}18.5\%
& \cellcolor{lightyellowRow}19.5\%
& \cellcolor{lightyellowRow}16.4\%
& \cellcolor{lightyellowRow}16.8\%
& \cellcolor{lightyellowRow}16.0\%
& \cellcolor{lightyellowRow}16.5\%
& \cellcolor{lightyellowRow}21.3\%
& \cellcolor{lightyellowRow}21.6\%
& \cellcolor{lightyellowRow}20.1\%
& \cellcolor{lightyellowRow}20.8\% \\
\bottomrule
\end{tabular}
\end{table*}

\subsection{Reinforcement Learning}
Building on the SFT checkpoint, we further fine-tune the model with RL so that its generated outputs better align with the final recommendation objective~\citep{ReRe}. This stage improves upon SFT in two ways. First, whereas SFT learns from imitation of ground-truth targets, RL optimizes the model directly in the actual decoding setting by learning from both successful and unsuccessful generations. Second, while SFT supervises the model with only one textual ID sequence per item, RL encourages the model to align with all valid textual ID variants of the same item, as we show later in \cref{subsec:RL ablation}. As a result, constrained beam search can preserve a larger set of correct decoding paths, which improves the overall success rate.

We use Group Relative Policy Optimization (GRPO) for this stage~\citep{deepseekmath,deepseek-R1}. GRPO improves the policy by sampling multiple rollouts for the same input and normalizing their rewards within the sampled group. Concretely, for each prompted user history $x$, we run constrained beam search over the trie to generate a set of candidate textual IDs:
$
\mathcal{Y}(x)=\{y^{(1)},y^{(2)},\ldots,y^{(G)}\},
$
where each candidate $y^{(g)}$ corresponds to a valid root-to-leaf path in the decoding trie. We then evaluate each rollout with a rule-based reward function that assigns a score according to whether the generated path matches the ground-truth next item:
\begin{equation}
R\bigl(y^{(g)}, Y^{*}\bigr)=
\begin{cases}
1, & y^{(g)} \in Y^{*},\\
0, & \text{otherwise},
\end{cases}
\end{equation}
where $Y^{*}$ is the set of target item term IDs. 
Let $\mu$ and $\sigma$ denote the mean and standard deviation of the group rewards -- the normalized advantage of the $g$-th rollout is computed as
$
\hat{A}^{(g)}=\frac{R^{(g)}-\mu}{\sigma+\epsilon}$,
where $\epsilon$ is a small constant employed for numerical stability.
We then update the policy with the GRPO objective as follows:
\begin{equation}
\mathcal{L}_{\mathrm{RL}}
=
-\frac{1}{G}\sum_{g=1}^{G}\sum_{k=1}^{|y^{(g)}|}
\min\Bigl(
\rho_{g,k}\hat{A}^{(g)},
\mathrm{clip}(\rho_{g,k},1-\eta,1+\eta)\hat{A}^{(g)}
\Bigr),
\end{equation}
where
$
\rho_{g,k}
=
\frac{P_{\theta}(y^{(g)}_{k}\mid x,y^{(g)}_{<k})}
{P_{\theta_{\mathrm{old}}}(y^{(g)}_{k}\mid x,y^{(g)}_{<k})}
$
is the token-level importance ratio, $\theta_{\mathrm{old}}$ is the policy before the update, and $\eta$ is the clipping coefficient.

This RL stage closes the critical gap between token-level imitation learning and item-level recommendation success. By optimizing the model directly within the constraints of the decoding trie, it becomes better equipped to utilize the trie-aware, variable-length representations.

\section{Experiments} \label{sec:experiment}
In this section, we will answer the following specific research questions with experiments.
\begin{itemize}[leftmargin=*]
\item \textbf{RQ1:} How does BONSAI compare with state-of-the-art traditional and GR methods in terms of recommendation accuracy?
\item \textbf{RQ2:} What is the contribution of each component in the two-stage training framework to overall model performance?
\item \textbf{RQ3:} How do BONSAI's trie-aware design choices (controlled branching factor, variable depth) affect decoding behavior and recommendation performance?
\item \textbf{RQ4:} Are BONSAI's structural design principles transferable to existing term-ID-based methods?
\item \textbf{RQ5:} To what extent does BONSAI's adaptive structure improve recommendation performance in cold-start settings for newly introduced items? ~(Results in \cref{app:RQ5})
\item \textbf{RQ6:} How does BONSAI scale with respect to backbone model size and training data volume?~(Results in \cref{app:RQ6})
\end{itemize}

\noindent \textbf{Baselines}: We evaluate BONSAI against a diverse set of traditional and state-of-the-art recommendation baselines, organized into four categories: traditional methods with atomic IDs, including SASRec~\citep{SASRec} and GRU4Rec~\citep{GRU4Rec}; train-from-scratch GR methods, including TIGER~\citep{TIGER} and HSTU~\citep{HSTU}; LLM-based GR methods with semantic IDs (SID), including MiniOneRec~\citep{minionerec} and OneRec-Think~\citep{onerec-think}; and LLM-based GR methods based on textual item representations, including IDGenRec~\citep{IDGenRec} and GRLM~\citep{TID}.

\noindent \textbf{Implementation Details}: We use Qwen3-32B as the LLM-based extractor for feature filtering. To reduce the computational cost of trie construction, we first partition the item set by metadata category and then apply the minimum set cover algorithm to each subset independently. \textit{To lower collision rate}, we append the first word of title to IDs with collisions~(a comparison of collision rate can be found in \cref{app:dataset}). For fair comparison across LLM-based methods, we standardize the backbone LLM as Qwen3-1.7B. In the SFT stage, we use an initial learning rate of $1\times10^{-4}$ with cosine decay. We then continue from the final SFT checkpoint for one epoch of RL using GRPO, with an initial learning rate of $1\times10^{-6}$ and a rollout beam size of 10. All models are trained on eight A100 80GB GPUs.

\noindent \textbf{Datasets and Metrics}: 
Consistent with prior research~\citep{onerec-think,TID}, evaluations are conducted on the Amazon Review datasets~\citep{SASRec}.
We use three subdatasets: {\small\textsf{Beauty}}, {\small\textsf{Sports and Outdoors}}, and {\small\textsf{Toys and Games}}. 
We also use Amazon-M2~\citep{AmazonM2} for data scaling experiments in \cref{app:RQ6}.
We use four metrics to evaluate the performance of the recommendation: Recall@5, Recall@10, NDCG@5 and NDCG@10.

\subsection{Overall Performance~(RQ1)}\label{subsec:overall_perf}
Table~\ref{tab:main_results} summarizes the overall performance. \textbf{BONSAI consistently and significantly outperforms all baselines across every dataset and evaluation metric}. Relative to the strongest baseline, GRLM~\citep{TID}, BONSAI achieves gains of \textbf{18.5\%--19.5\%} on {\small\textsf{Beauty}}, \textbf{16.0\%--16.8\%} on {\small\textsf{Sports}}, and \textbf{20.1\%--21.6\%} on {\small\textsf{Toys}}.

The results suggest three main observations. First, conventional sequential recommendation models, including SASRec and GRU4Rec, perform the worst overall. This is unsurprising, as they rely on atomic item IDs and therefore cannot leverage the semantic knowledge encoded in pretrained LLMs. Second, generative recommendation models trained from scratch, such as TIGER and HSTU, are more competitive than traditional methods in some settings, but still underperform LLM-based approaches. This indicates that the open-world knowledge captured by pretrained LLM backbones is highly beneficial for generative recommendation. Third, among LLM-based baselines, methods that use textual item representations, such as IDGenRec and GRLM, generally outperform SID-based methods. This is consistent with recent findings that textual representations better preserve item semantics and align more naturally with the native LLM vocabulary~\citep{liu2025understanding}.



\subsection{Contribution of Two-Stage Training~(RQ2)}\label{subsec:RL ablation}

We next examine the contribution of the proposed two-stage training strategy, which combines supervised fine-tuning (SFT) with reinforcement learning (RL). The results are shown in  Figure~\ref{fig:RL_impr}, where the RL improvement is calculated as 
$$
    \text{RL improvement}=\frac{\text{Recall@5 after RL}-\text{Recall@5 after SFT}}{\text{Recall@5 after SFT}}
$$
\noindent The results show that RL yields consistent gains over SFT alone across all datasets, with relative improvements of roughly \textbf{6\%--10\%}. This suggests that \textbf{SFT already provides strong performance which surpasses all baselines in \cref{tab:main_results}, while RL further improves performance consistently on top of it}. 

Furthermore, the gains are mostly from items with long representations~(\textbf{Top 50\% Long Repr}) and items with multiple valid representations~(\textbf{Multi-Repr}). Notably, these subsets overlap substantially, particularly in \textbf{Multi-Repr+Top 50\% Long}. This pattern is closely aligned with our design motivation: RL helps the model better handle the multiple valid representations associated with an item.

Overall, these results highlight the importance of RL in our framework. By bridging the gap between token-level imitation learning and item-level recommendation performance, RL improves the model's ability to handle the more challenging cases introduced by trie-aware variable-length representations. The two-stage training procedure is therefore well matched to the structure of BONSAI and is a key factor in its strong empirical performance.

\begin{figure}[t]
\begin{center}

\includegraphics[width=0.35\textwidth]{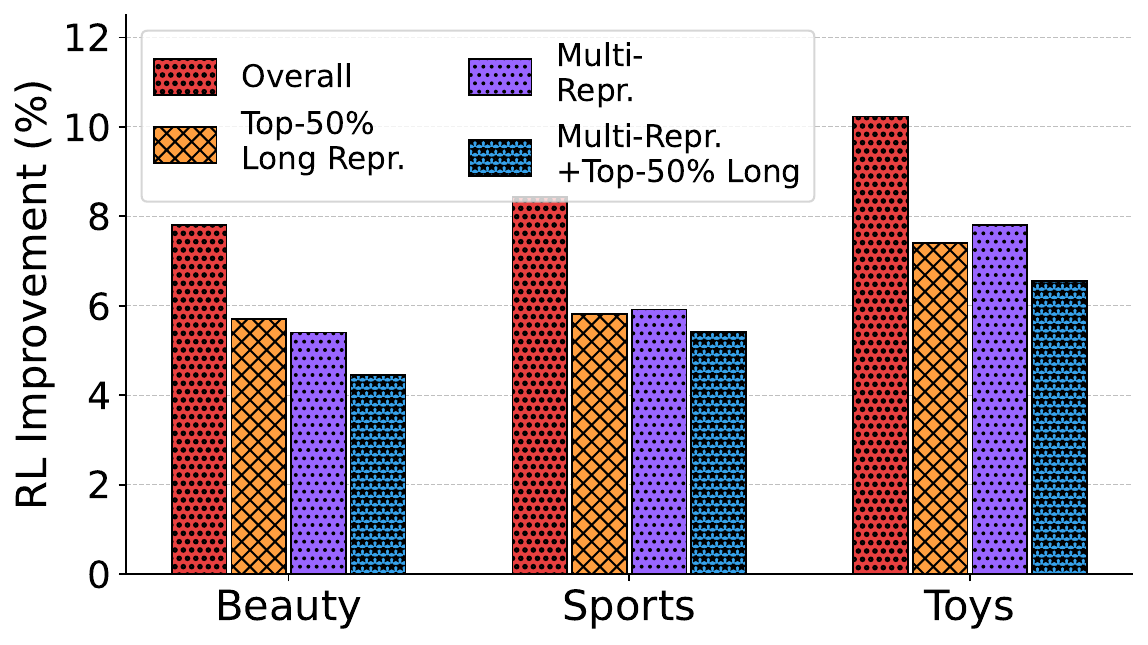}
\vspace{-1.5em}
\caption{The improvements from RL over SFT. Most improvement comes from the items with long or multiple IDs.}
\vspace{-1.5em}
\label{fig:RL_impr}
\end{center}

\end{figure}

\begin{figure}[t]
\begin{center}

\includegraphics[width=0.5\textwidth]{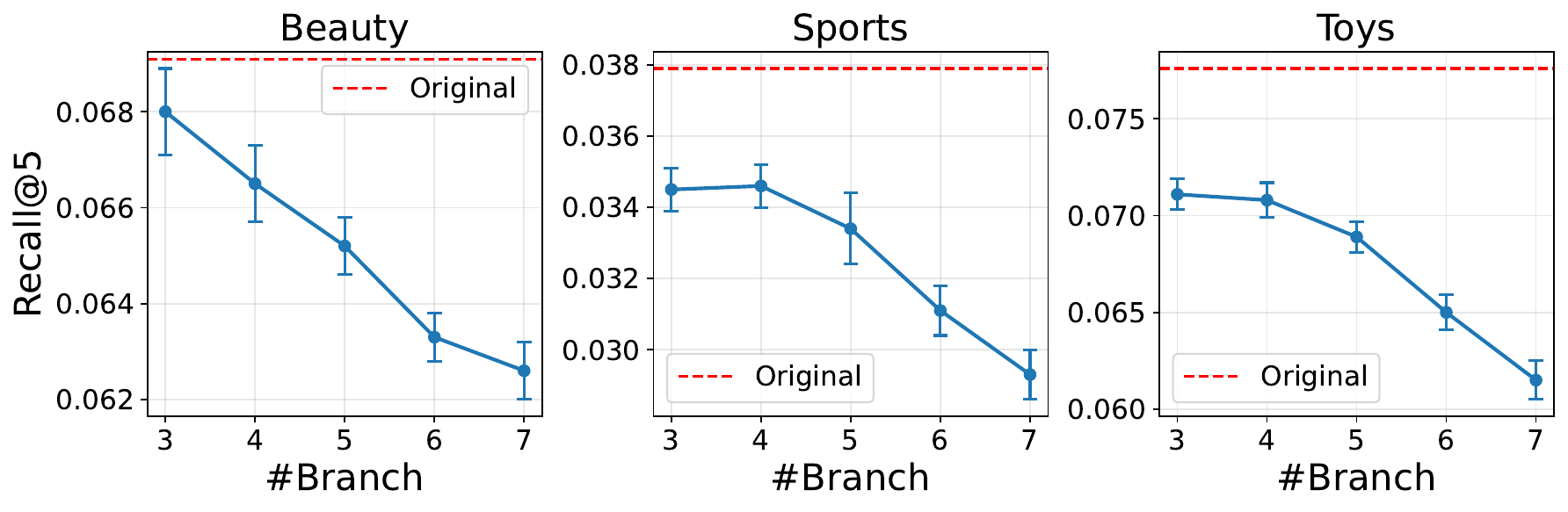}
\vspace{-2em}
\caption{Impacts of varying branching factor. The performance drops when the branching factor increases in general. }
\label{fig:abla_branch}
\vspace{-1.5em}
\end{center}
\end{figure}

\begin{figure}[t]
\begin{center}
\includegraphics[width=0.5\textwidth]{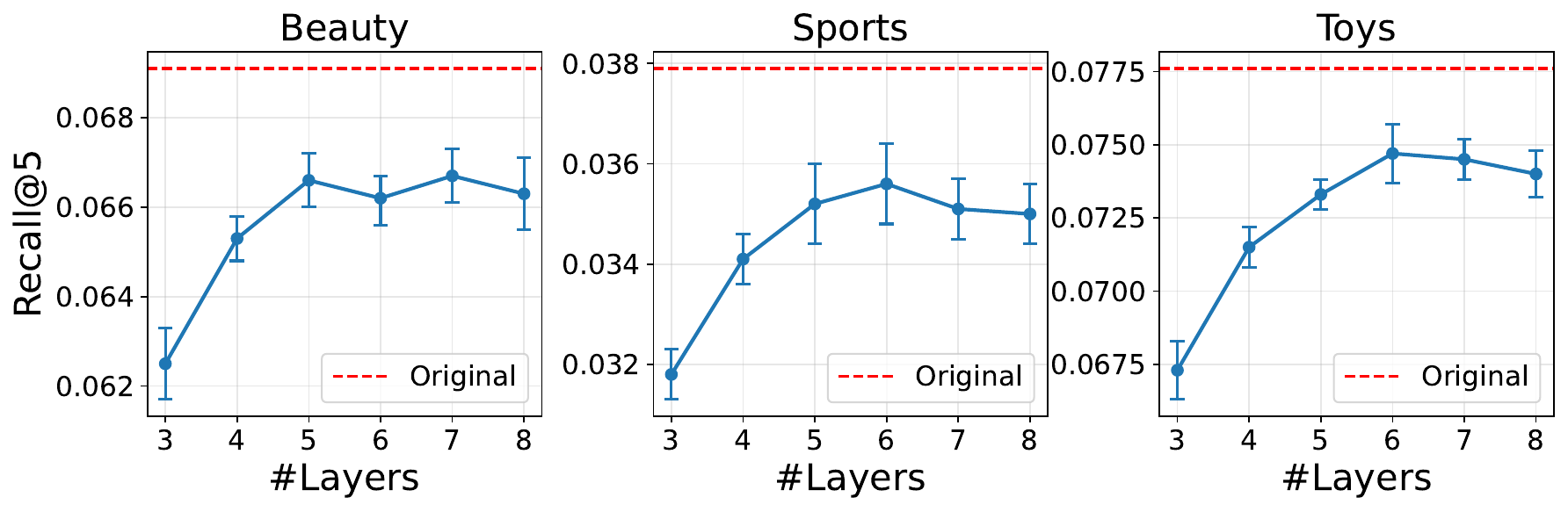}
\vspace{-2em}
\caption{Impacts of trie depth. We observe consistent gaps between the uniform-depth trie and the original one.}
\label{fig:abla_depth}
\vspace{-1.5em}
\end{center}
\end{figure}

\subsection{Contribution of Design Properties~(RQ3)}\label{subsec:trie_ablation}
To verify that the proposed trie-aware properties drive the performance gains, we conduct controlled ablations on two key factors: branching factor and variable depth. The results, shown in Figure~\ref{fig:abla_branch} and Figure~\ref{fig:abla_depth}, support our claim that \textbf{both a controlled branching factor and variable depth are crucial for strong LLM-based generative recommendation performance}.

We first study the effect of branching factor in Figure~\ref{fig:abla_branch}. In this experiment, we fix the trie depth to 4 and vary the branching factors of the intermediate layers by replacing the minimum-cover construction with a fixed-size cover-set scheme during trie construction. For example, to enforce a branching factor of 5, we select five features to partition the item subset at each level. At the final layer, we append unique suffixes to items assigned to the same leaf to avoid representation collisions. Across all three datasets, we observe a clear and consistent trend: \textit{recommendation performance degrades steadily as the branching factor increases}. This supports our core intuition that early decisions in constrained beam search are especially critical. When too many branches are available in shallow layers, the model is less likely to select the correct path, and an early mistake can remove the target item from the beam entirely. In contrast, our minimum-cover-based construction keeps the branching factor small, substantially simplifying search.

We next examine the effect of trie depth in Figure~\ref{fig:abla_depth} by truncating the original trie to produce shallower, more uniform structures. Performance generally improves as the trie becomes deeper, but the original trie consistently outperforms all truncated variants. This finding suggests that imposing a fixed-depth structure on all items is suboptimal. Items differ naturally in semantic richness and in the amount of information needed to distinguish them. A uniform depth therefore either over-compresses semantically rich items or imposes unnecessary structure on simpler ones. In contrast, BONSAI allows more semantically complex items to occupy deeper paths while assigning shorter representations to simpler items, yielding a better trade-off between representational fidelity and decoding difficulty.

\subsection{Transferability to Existing Methods~(RQ4)}\label{subsec:optimize_tid}

We next examine whether our proposed principles generalize beyond BONSAI. To do so, we reorganize the term IDs generated by a leading LLM-based GR method, GRLM~\citep{TID}, and report the results in Table~\ref{tab:TID_impr}. Specifically, we preserve the original keywords in each ID and reorder them using the minimum-cover-set strategy. The results show that \textbf{our strategy transfers effectively to other methods and yields consistent improvements}.

Specifically, Table~\ref{tab:TID_impr} shows that reorganizing the original term IDs makes the branching factors substantially more balanced across trie levels on all three datasets. In the original organization, the first trie level has an extremely large branching factor, making the initial decoding decision particularly difficult. Once the wrong branch is chosen at the first level, the remaining decoding steps can no longer recover the correct item. After applying our method, the first-level branching factor is reduced dramatically, while deeper levels become moderately wider in a more controlled way. This reorganization leads to consistent performance gains over the original results reported in \cref{tab:main_results}. These findings further support our claim that the decoding trie is a major performance bottleneck. Even when item semantics remain unchanged, simply reorganizing the same information into a more suitable trie can noticeably improve performance.



\begin{table}[t]
\centering
\caption{Experiments on reorganizing item term IDs from GRLM~\citep{TID}. "O" stands for the original trie, "M" stands for the modified trie with the minimum cover set. "Impr." indicates the performance improvement after reorganizing.}
\vspace{-1em}
\label{tab:TID_impr}
\small
\setlength{\tabcolsep}{2.5pt}
\renewcommand{\arraystretch}{1.2}
\begin{tabular}{l |c|ccccc|c}
\toprule
\multirow{2}{*}{\textbf{Dataset}} & \multirow{2}{*}{\textbf{Method}} & \multicolumn{5}{c|}{\textbf{Average Branching Factor}} & \multirow{2}{*}{\textbf{Impr.}} \\
\cline{3-7}
& & \textbf{Level 1} & \textbf{Level 2} & \textbf{Level 3} & \textbf{Level 4} & \textbf{Level 5} & \\
\midrule
\multirow{2}{*}{\textbf{Beauty}}
& O & 1114 & 5.39  & 1.70 & 1.12 & 1.04 & \multirow{2}{*}{8.9\%} \\
& M & 375  & 15.80 & 3.72 & 2.05 & 1.65 & \\
\midrule
\multirow{2}{*}{\textbf{Sports}}
& O & 2344 & 4.29  & 1.54 & 1.11 & 1.04 & \multirow{2}{*}{7.7\%} \\
& M & 508  & 15.32 & 3.59 & 1.90 & 1.57 & \\
\midrule
\multirow{2}{*}{\textbf{Toys}}
& O & 2085 & 3.75  & 1.39 & 1.06 & 1.02 & \multirow{2}{*}{9.1\%} \\
& M & 373  & 15.89 & 4.55 & 2.60 & 1.74 & \\
\bottomrule
\end{tabular}
\end{table}

\section{Conclusion}\label{sec:Conlusions}

In this work, we propose a new perspective for designing item IDs in generative recommendations. Our results suggest that future research should not treat item representation quality and decoding structure as separate problems; instead, they should be optimized jointly. 
There are still many topics to explore in this direction.
The possible future works include improving feature extraction methods, achieving better trade-offs between trie depth and width, and extending these principles to multi-modality item representations.

\bibliographystyle{ACM-Reference-Format}
\bibliography{main}


\begin{thebibliography}{56}


\ifx \showCODEN    \undefined \def \showCODEN     #1{\unskip}     \fi
\ifx \showISBNx    \undefined \def \showISBNx     #1{\unskip}     \fi
\ifx \showISBNxiii \undefined \def \showISBNxiii  #1{\unskip}     \fi
\ifx \showISSN     \undefined \def \showISSN      #1{\unskip}     \fi
\ifx \showLCCN     \undefined \def \showLCCN      #1{\unskip}     \fi
\ifx \shownote     \undefined \def \shownote      #1{#1}          \fi
\ifx \showarticletitle \undefined \def \showarticletitle #1{#1}   \fi
\ifx \showURL      \undefined \def \showURL       {\relax}        \fi
\providecommand\bibfield[2]{#2}
\providecommand\bibinfo[2]{#2}
\providecommand\natexlab[1]{#1}
\providecommand\showeprint[2][]{arXiv:#2}

\bibitem[Bao et~al\mbox{.}(2024)]%
        {decodingmatters}
\bibfield{author}{\bibinfo{person}{Keqin Bao}, \bibinfo{person}{Jizhi Zhang}, \bibinfo{person}{Yang Zhang}, \bibinfo{person}{Xinyue Huo}, \bibinfo{person}{Chong Chen}, {and} \bibinfo{person}{Fuli Feng}.} \bibinfo{year}{2024}\natexlab{}.
\newblock \showarticletitle{Decoding matters: Addressing amplification bias and homogeneity issue for llm-based recommendation}.
\newblock \bibinfo{journal}{\emph{arXiv preprint arXiv:2406.14900}} (\bibinfo{year}{2024}).
\newblock


\bibitem[Bao et~al\mbox{.}(2023)]%
        {TALLRec}
\bibfield{author}{\bibinfo{person}{Keqin Bao}, \bibinfo{person}{Jizhi Zhang}, \bibinfo{person}{Yang Zhang}, \bibinfo{person}{Wenjie Wang}, \bibinfo{person}{Fuli Feng}, {and} \bibinfo{person}{Xiangnan He}.} \bibinfo{year}{2023}\natexlab{}.
\newblock \showarticletitle{TALLRec: An Effective and Efficient Tuning Framework to Align Large Language Model with Recommendation}. In \bibinfo{booktitle}{\emph{RecSys}}.
\newblock


\bibitem[Bertsekas(2012)]%
        {bertsekas2012dynamic}
\bibfield{author}{\bibinfo{person}{Dimitri Bertsekas}.} \bibinfo{year}{2012}\natexlab{}.
\newblock \bibinfo{booktitle}{\emph{Dynamic programming and optimal control: Volume I}}. Vol.~\bibinfo{volume}{4}.
\newblock \bibinfo{publisher}{Athena scientific}.
\newblock


\bibitem[Chen et~al\mbox{.}(2024)]%
        {chen2024enhancing}
\bibfield{author}{\bibinfo{person}{Runjin Chen}, \bibinfo{person}{Mingxuan Ju}, \bibinfo{person}{Ngoc Bui}, \bibinfo{person}{Dimosthenis Antypas}, \bibinfo{person}{Stanley Cai}, \bibinfo{person}{Xiaopeng Wu}, \bibinfo{person}{Leonardo Neves}, \bibinfo{person}{Zhangyang Wang}, \bibinfo{person}{Neil Shah}, {and} \bibinfo{person}{Tong Zhao}.} \bibinfo{year}{2024}\natexlab{}.
\newblock \showarticletitle{Enhancing item tokenization for generative recommendation through self-improvement}.
\newblock \bibinfo{journal}{\emph{arXiv preprint arXiv:2412.17171}} (\bibinfo{year}{2024}).
\newblock


\bibitem[Chen et~al\mbox{.}(2026)]%
        {memrec}
\bibfield{author}{\bibinfo{person}{Weixin Chen}, \bibinfo{person}{Yuhan Zhao}, \bibinfo{person}{Jingyuan Huang}, \bibinfo{person}{Zihe Ye}, \bibinfo{person}{Clark~Mingxuan Ju}, \bibinfo{person}{Tong Zhao}, \bibinfo{person}{Neil Shah}, \bibinfo{person}{Li Chen}, {and} \bibinfo{person}{Yongfeng Zhang}.} \bibinfo{year}{2026}\natexlab{}.
\newblock \showarticletitle{MemRec: Collaborative Memory-Augmented Agentic Recommender System}.
\newblock \bibinfo{journal}{\emph{arXiv preprint arXiv:2601.08816}} (\bibinfo{year}{2026}).
\newblock


\bibitem[Covington et~al\mbox{.}(2016)]%
        {covington2016deep}
\bibfield{author}{\bibinfo{person}{Paul Covington}, \bibinfo{person}{Jay Adams}, {and} \bibinfo{person}{Emre Sargin}.} \bibinfo{year}{2016}\natexlab{}.
\newblock \showarticletitle{Deep neural networks for youtube recommendations}. In \bibinfo{booktitle}{\emph{Proceedings of the 10th ACM conference on recommender systems}}. \bibinfo{pages}{191--198}.
\newblock


\bibitem[Deng et~al\mbox{.}(2025)]%
        {OneRec}
\bibfield{author}{\bibinfo{person}{Jiaxin Deng}, \bibinfo{person}{Shiyao Wang}, \bibinfo{person}{Kuo Cai}, \bibinfo{person}{Lejian Ren}, \bibinfo{person}{Qigen Hu}, \bibinfo{person}{Weifeng Ding}, \bibinfo{person}{Qiang Luo}, {and} \bibinfo{person}{Guorui Zhou}.} \bibinfo{year}{2025}\natexlab{}.
\newblock \showarticletitle{Onerec: Unifying retrieve and rank with generative recommender and iterative preference alignment}.
\newblock \bibinfo{journal}{\emph{arXiv preprint arXiv:2502.18965}} (\bibinfo{year}{2025}).
\newblock


\bibitem[Ding et~al\mbox{.}(2026)]%
        {ultra-hstu}
\bibfield{author}{\bibinfo{person}{Qin Ding}, \bibinfo{person}{Kevin Course}, \bibinfo{person}{Linjian Ma}, \bibinfo{person}{Jianhui Sun}, \bibinfo{person}{Rouchen Liu}, \bibinfo{person}{Zhao Zhu}, \bibinfo{person}{Chunxing Yin}, \bibinfo{person}{Wei Li}, \bibinfo{person}{Dai Li}, \bibinfo{person}{Yu Shi}, {et~al\mbox{.}}} \bibinfo{year}{2026}\natexlab{}.
\newblock \showarticletitle{Bending the Scaling Law Curve in Large-Scale Recommendation Systems}.
\newblock \bibinfo{journal}{\emph{arXiv preprint arXiv:2602.16986}} (\bibinfo{year}{2026}).
\newblock


\bibitem[Geng et~al\mbox{.}(2022)]%
        {P5}
\bibfield{author}{\bibinfo{person}{Shijie Geng}, \bibinfo{person}{Shuchang Liu}, \bibinfo{person}{Zuohui Fu}, \bibinfo{person}{Yingqiang Ge}, {and} \bibinfo{person}{Yongfeng Zhang}.} \bibinfo{year}{2022}\natexlab{}.
\newblock \showarticletitle{Recommendation as Language Processing {(RLP):} {A} Unified Pretrain, Personalized Prompt {\&} Predict Paradigm {(P5)}}. In \bibinfo{booktitle}{\emph{RecSys}}.
\newblock


\bibitem[Guo et~al\mbox{.}(2025)]%
        {deepseek-R1}
\bibfield{author}{\bibinfo{person}{Daya Guo}, \bibinfo{person}{Dejian Yang}, \bibinfo{person}{Haowei Zhang}, \bibinfo{person}{Junxiao Song}, \bibinfo{person}{Peiyi Wang}, \bibinfo{person}{Qihao Zhu}, \bibinfo{person}{Runxin Xu}, \bibinfo{person}{Ruoyu Zhang}, \bibinfo{person}{Shirong Ma}, \bibinfo{person}{Xiao Bi}, {et~al\mbox{.}}} \bibinfo{year}{2025}\natexlab{}.
\newblock \showarticletitle{Deepseek-r1: Incentivizing reasoning capability in llms via reinforcement learning}.
\newblock \bibinfo{journal}{\emph{arXiv preprint arXiv:2501.12948}} (\bibinfo{year}{2025}).
\newblock


\bibitem[He et~al\mbox{.}(2025)]%
        {plum}
\bibfield{author}{\bibinfo{person}{Ruining He}, \bibinfo{person}{Lukasz Heldt}, \bibinfo{person}{Lichan Hong}, \bibinfo{person}{Raghunandan Keshavan}, \bibinfo{person}{Shifan Mao}, \bibinfo{person}{Nikhil Mehta}, \bibinfo{person}{Zhengyang Su}, \bibinfo{person}{Alicia Tsai}, \bibinfo{person}{Yueqi Wang}, \bibinfo{person}{Shao-Chuan Wang}, {et~al\mbox{.}}} \bibinfo{year}{2025}\natexlab{}.
\newblock \showarticletitle{Plum: Adapting pre-trained language models for industrial-scale generative recommendations}.
\newblock \bibinfo{journal}{\emph{arXiv preprint arXiv:2510.07784}} (\bibinfo{year}{2025}).
\newblock


\bibitem[Hong et~al\mbox{.}(2025)]%
        {Gream}
\bibfield{author}{\bibinfo{person}{Minjie Hong}, \bibinfo{person}{Zetong Zhou}, \bibinfo{person}{Zirun Guo}, \bibinfo{person}{Ziang Zhang}, \bibinfo{person}{Ruofan Hu}, \bibinfo{person}{Weinan Gan}, \bibinfo{person}{Jieming Zhu}, {and} \bibinfo{person}{Zhou Zhao}.} \bibinfo{year}{2025}\natexlab{}.
\newblock \showarticletitle{Generative Reasoning Recommendation via LLMs}.
\newblock \bibinfo{journal}{\emph{arXiv preprint arXiv:2510.20815}} (\bibinfo{year}{2025}).
\newblock


\bibitem[Hou et~al\mbox{.}(2025)]%
        {hou2025actionpiece}
\bibfield{author}{\bibinfo{person}{Yupeng Hou}, \bibinfo{person}{Jianmo Ni}, \bibinfo{person}{Zhankui He}, \bibinfo{person}{Noveen Sachdeva}, \bibinfo{person}{Wang-Cheng Kang}, \bibinfo{person}{Ed~H Chi}, \bibinfo{person}{Julian McAuley}, {and} \bibinfo{person}{Derek~Zhiyuan Cheng}.} \bibinfo{year}{2025}\natexlab{}.
\newblock \showarticletitle{ActionPiece: Contextually Tokenizing Action Sequences for Generative Recommendation}.
\newblock \bibinfo{journal}{\emph{arXiv preprint arXiv:2502.13581}} (\bibinfo{year}{2025}).
\newblock


\bibitem[Hu et~al\mbox{.}(2026)]%
        {hu2026ids}
\bibfield{author}{\bibinfo{person}{Peiyu Hu}, \bibinfo{person}{Wayne Lu}, {and} \bibinfo{person}{Jia Wang}.} \bibinfo{year}{2026}\natexlab{}.
\newblock \showarticletitle{From ids to semantics: A generative framework for cross-domain recommendation with adaptive semantic tokenization}. In \bibinfo{booktitle}{\emph{Proceedings of the AAAI Conference on Artificial Intelligence}}, Vol.~\bibinfo{volume}{40}. \bibinfo{pages}{14874--14882}.
\newblock


\bibitem[Hua et~al\mbox{.}(2023)]%
        {how-to-index}
\bibfield{author}{\bibinfo{person}{Wenyue Hua}, \bibinfo{person}{Shuyuan Xu}, \bibinfo{person}{Yingqiang Ge}, {and} \bibinfo{person}{Yongfeng Zhang}.} \bibinfo{year}{2023}\natexlab{}.
\newblock \showarticletitle{How to Index Item IDs for Recommendation Foundation Models}. In \bibinfo{booktitle}{\emph{SIGIR-AP}}, \bibfield{editor}{\bibinfo{person}{Qingyao Ai}, \bibinfo{person}{Yiqin Liu}, \bibinfo{person}{Alistair Moffat}, \bibinfo{person}{Xuanjing Huang}, \bibinfo{person}{Tetsuya Sakai}, {and} \bibinfo{person}{Justin Zobel}} (Eds.).
\newblock


\bibitem[Huang et~al\mbox{.}(2020)]%
        {huang2020embedding}
\bibfield{author}{\bibinfo{person}{Jui-Ting Huang}, \bibinfo{person}{Ashish Sharma}, \bibinfo{person}{Shuying Sun}, \bibinfo{person}{Li Xia}, \bibinfo{person}{David Zhang}, \bibinfo{person}{Philip Pronin}, \bibinfo{person}{Janani Padmanabhan}, \bibinfo{person}{Giuseppe Ottaviano}, {and} \bibinfo{person}{Linjun Yang}.} \bibinfo{year}{2020}\natexlab{}.
\newblock \showarticletitle{Embedding-based retrieval in facebook search}. In \bibinfo{booktitle}{\emph{Proceedings of the 26th ACM SIGKDD International Conference on Knowledge Discovery \& Data Mining}}. \bibinfo{pages}{2553--2561}.
\newblock


\bibitem[Jannach and Ludewig(2017)]%
        {GRU4Rec}
\bibfield{author}{\bibinfo{person}{Dietmar Jannach} {and} \bibinfo{person}{Malte Ludewig}.} \bibinfo{year}{2017}\natexlab{}.
\newblock \showarticletitle{When Recurrent Neural Networks meet the Neighborhood for Session-Based Recommendation}. In \bibinfo{booktitle}{\emph{RecSys}}, \bibfield{editor}{\bibinfo{person}{Paolo Cremonesi}, \bibinfo{person}{Francesco Ricci}, \bibinfo{person}{Shlomo Berkovsky}, {and} \bibinfo{person}{Alexander Tuzhilin}} (Eds.). \bibinfo{publisher}{ACM}.
\newblock


\bibitem[Ji et~al\mbox{.}(2024)]%
        {GenRec}
\bibfield{author}{\bibinfo{person}{Jianchao Ji}, \bibinfo{person}{Zelong Li}, \bibinfo{person}{Shuyuan Xu}, \bibinfo{person}{Wenyue Hua}, \bibinfo{person}{Yingqiang Ge}, \bibinfo{person}{Juntao Tan}, {and} \bibinfo{person}{Yongfeng Zhang}.} \bibinfo{year}{2024}\natexlab{}.
\newblock \showarticletitle{Genrec: Large language model for generative recommendation}. In \bibinfo{booktitle}{\emph{European Conference on Information Retrieval}}. Springer, \bibinfo{pages}{494--502}.
\newblock


\bibitem[Jin et~al\mbox{.}(2023)]%
        {AmazonM2}
\bibfield{author}{\bibinfo{person}{Wei Jin}, \bibinfo{person}{Haitao Mao}, \bibinfo{person}{Zheng Li}, \bibinfo{person}{Haoming Jiang}, \bibinfo{person}{Chen Luo}, \bibinfo{person}{Hongzhi Wen}, \bibinfo{person}{Haoyu Han}, \bibinfo{person}{Hanqing Lu}, \bibinfo{person}{Zhengyang Wang}, \bibinfo{person}{Ruirui Li}, {et~al\mbox{.}}} \bibinfo{year}{2023}\natexlab{}.
\newblock \showarticletitle{Amazon-m2: A multilingual multi-locale shopping session dataset for recommendation and text generation}.
\newblock \bibinfo{journal}{\emph{Advances in Neural Information Processing Systems}}  \bibinfo{volume}{36} (\bibinfo{year}{2023}), \bibinfo{pages}{8006--8026}.
\newblock


\bibitem[Ju et~al\mbox{.}(2025a)]%
        {ju2025generative}
\bibfield{author}{\bibinfo{person}{Clark~Mingxuan Ju}, \bibinfo{person}{Liam Collins}, \bibinfo{person}{Leonardo Neves}, \bibinfo{person}{Bhuvesh Kumar}, \bibinfo{person}{Louis~Yufeng Wang}, \bibinfo{person}{Tong Zhao}, {and} \bibinfo{person}{Neil Shah}.} \bibinfo{year}{2025}\natexlab{a}.
\newblock \showarticletitle{Generative Recommendation with Semantic IDs: A Practitioner's Handbook}.
\newblock \bibinfo{journal}{\emph{arXiv preprint arXiv:2507.22224}} (\bibinfo{year}{2025}).
\newblock


\bibitem[Ju et~al\mbox{.}(2025b)]%
        {ju2025revisiting}
\bibfield{author}{\bibinfo{person}{Clark~Mingxuan Ju}, \bibinfo{person}{Leonardo Neves}, \bibinfo{person}{Bhuvesh Kumar}, \bibinfo{person}{Liam Collins}, \bibinfo{person}{Tong Zhao}, \bibinfo{person}{Yuwei Qiu}, \bibinfo{person}{Qing Dou}, \bibinfo{person}{Sohail Nizam}, \bibinfo{person}{Sen Yang}, {and} \bibinfo{person}{Neil Shah}.} \bibinfo{year}{2025}\natexlab{b}.
\newblock \showarticletitle{Revisiting Self-attention for Cross-domain Sequential Recommendation}. In \bibinfo{booktitle}{\emph{Proceedings of the 31st ACM SIGKDD Conference on Knowledge Discovery and Data Mining V. 2}}. \bibinfo{pages}{1094--1105}.
\newblock


\bibitem[Ju et~al\mbox{.}(2025c)]%
        {ju2025learning}
\bibfield{author}{\bibinfo{person}{Clark~Mingxuan Ju}, \bibinfo{person}{Leonardo Neves}, \bibinfo{person}{Bhuvesh Kumar}, \bibinfo{person}{Liam Collins}, \bibinfo{person}{Tong Zhao}, \bibinfo{person}{Yuwei Qiu}, \bibinfo{person}{Qing Dou}, \bibinfo{person}{Yang Zhou}, \bibinfo{person}{Sohail Nizam}, \bibinfo{person}{Rengim Ozturk}, {et~al\mbox{.}}} \bibinfo{year}{2025}\natexlab{c}.
\newblock \showarticletitle{Learning Universal User Representations Leveraging Cross-domain User Intent at Snapchat}.
\newblock \bibinfo{journal}{\emph{arXiv preprint arXiv:2504.21838}} (\bibinfo{year}{2025}).
\newblock


\bibitem[Ju et~al\mbox{.}(2026)]%
        {ju2026semantic}
\bibfield{author}{\bibinfo{person}{Clark~Mingxuan Ju}, \bibinfo{person}{Tong Zhao}, \bibinfo{person}{Leonardo Neves}, \bibinfo{person}{Liam Collins}, \bibinfo{person}{Bhuvesh Kumar}, \bibinfo{person}{Jiwen Ren}, \bibinfo{person}{Lili Zhang}, \bibinfo{person}{Wenfeng Zhuo}, \bibinfo{person}{Vincent Zhang}, \bibinfo{person}{Xiao Bai}, {et~al\mbox{.}}} \bibinfo{year}{2026}\natexlab{}.
\newblock \showarticletitle{Semantic IDs for Recommender Systems at Snapchat: Use Cases, Technical Challenges, and Design Choices}.
\newblock \bibinfo{journal}{\emph{arXiv preprint arXiv:2604.03949}} (\bibinfo{year}{2026}).
\newblock


\bibitem[Ju et~al\mbox{.}(2024)]%
        {ju2024does}
\bibfield{author}{\bibinfo{person}{Mingxuan Ju}, \bibinfo{person}{William Shiao}, \bibinfo{person}{Zhichun Guo}, \bibinfo{person}{Yanfang Ye}, \bibinfo{person}{Yozen Liu}, \bibinfo{person}{Neil Shah}, {and} \bibinfo{person}{Tong Zhao}.} \bibinfo{year}{2024}\natexlab{}.
\newblock \showarticletitle{How does message passing improve collaborative filtering?}
\newblock \bibinfo{journal}{\emph{Advances in Neural Information Processing Systems}}  \bibinfo{volume}{37} (\bibinfo{year}{2024}), \bibinfo{pages}{8760--8784}.
\newblock


\bibitem[Kang and McAuley(2018)]%
        {SASRec}
\bibfield{author}{\bibinfo{person}{Wang{-}Cheng Kang} {and} \bibinfo{person}{Julian~J. McAuley}.} \bibinfo{year}{2018}\natexlab{}.
\newblock \showarticletitle{Self-Attentive Sequential Recommendation}. In \bibinfo{booktitle}{\emph{ICDM}}.
\newblock


\bibitem[Kong et~al\mbox{.}(2025)]%
        {minionerec}
\bibfield{author}{\bibinfo{person}{Xiaoyu Kong}, \bibinfo{person}{Leheng Sheng}, \bibinfo{person}{Junfei Tan}, \bibinfo{person}{Yuxin Chen}, \bibinfo{person}{Jiancan Wu}, \bibinfo{person}{An Zhang}, \bibinfo{person}{Xiang Wang}, {and} \bibinfo{person}{Xiangnan He}.} \bibinfo{year}{2025}\natexlab{}.
\newblock \showarticletitle{Minionerec: An open-source framework for scaling generative recommendation}.
\newblock \bibinfo{journal}{\emph{arXiv preprint arXiv:2510.24431}} (\bibinfo{year}{2025}).
\newblock


\bibitem[Lee et~al\mbox{.}(2022)]%
        {RQ22}
\bibfield{author}{\bibinfo{person}{Doyup Lee}, \bibinfo{person}{Chiheon Kim}, \bibinfo{person}{Saehoon Kim}, \bibinfo{person}{Minsu Cho}, {and} \bibinfo{person}{Wook-Shin Han}.} \bibinfo{year}{2022}\natexlab{}.
\newblock \showarticletitle{Autoregressive image generation using residual quantization}. In \bibinfo{booktitle}{\emph{Proceedings of the IEEE/CVF conference on computer vision and pattern recognition}}. \bibinfo{pages}{11523--11532}.
\newblock


\bibitem[Li et~al\mbox{.}(2023)]%
        {LLM4GenRec}
\bibfield{author}{\bibinfo{person}{Lei Li}, \bibinfo{person}{Yongfeng Zhang}, \bibinfo{person}{Dugang Liu}, {and} \bibinfo{person}{Li Chen}.} \bibinfo{year}{2023}\natexlab{}.
\newblock \showarticletitle{Large Language Models for Generative Recommendation: {A} Survey and Visionary Discussions}.
\newblock \bibinfo{journal}{\emph{CoRR}}  \bibinfo{volume}{abs/2309.01157} (\bibinfo{year}{2023}).
\newblock


\bibitem[Liao et~al\mbox{.}(2024)]%
        {LLaRA}
\bibfield{author}{\bibinfo{person}{Jiayi Liao}, \bibinfo{person}{Sihang Li}, \bibinfo{person}{Zhengyi Yang}, \bibinfo{person}{Jiancan Wu}, \bibinfo{person}{Yancheng Yuan}, \bibinfo{person}{Xiang Wang}, {and} \bibinfo{person}{Xiangnan He}.} \bibinfo{year}{2024}\natexlab{}.
\newblock \showarticletitle{LLaRA: Large Language-Recommendation Assistant}. In \bibinfo{booktitle}{\emph{SIGIR}}.
\newblock


\bibitem[Liu et~al\mbox{.}(2025a)]%
        {liu2025understanding}
\bibfield{author}{\bibinfo{person}{Jingzhe Liu}, \bibinfo{person}{Liam Collins}, \bibinfo{person}{Jiliang Tang}, \bibinfo{person}{Tong Zhao}, \bibinfo{person}{Neil Shah}, {and} \bibinfo{person}{Clark~Mingxuan Ju}.} \bibinfo{year}{2025}\natexlab{a}.
\newblock \showarticletitle{Understanding generative recommendation with semantic ids from a model-scaling view}.
\newblock \bibinfo{journal}{\emph{arXiv preprint arXiv:2509.25522}} (\bibinfo{year}{2025}).
\newblock


\bibitem[Liu et~al\mbox{.}(2025c)]%
        {liu2025bridge}
\bibfield{author}{\bibinfo{person}{Qidong Liu}, \bibinfo{person}{Xiangyu Zhao}, \bibinfo{person}{Yejing Wang}, \bibinfo{person}{Zijian Zhang}, \bibinfo{person}{Howard Zhong}, \bibinfo{person}{Chong Chen}, \bibinfo{person}{Xiang Li}, \bibinfo{person}{Wei Huang}, {and} \bibinfo{person}{Feng Tian}.} \bibinfo{year}{2025}\natexlab{c}.
\newblock \showarticletitle{Bridge the domains: Large language models enhanced cross-domain sequential recommendation}. In \bibinfo{booktitle}{\emph{Proceedings of the 48th International ACM SIGIR Conference on Research and Development in Information Retrieval}}. \bibinfo{pages}{1582--1592}.
\newblock


\bibitem[Liu et~al\mbox{.}(2025b)]%
        {onerec-think}
\bibfield{author}{\bibinfo{person}{Zhanyu Liu}, \bibinfo{person}{Shiyao Wang}, \bibinfo{person}{Xingmei Wang}, \bibinfo{person}{Rongzhou Zhang}, \bibinfo{person}{Jiaxin Deng}, \bibinfo{person}{Honghui Bao}, \bibinfo{person}{Jinghao Zhang}, \bibinfo{person}{Wuchao Li}, \bibinfo{person}{Pengfei Zheng}, \bibinfo{person}{Xiangyu Wu}, {et~al\mbox{.}}} \bibinfo{year}{2025}\natexlab{b}.
\newblock \showarticletitle{Onerec-think: In-text reasoning for generative recommendation}.
\newblock \bibinfo{journal}{\emph{arXiv preprint arXiv:2510.11639}} (\bibinfo{year}{2025}).
\newblock


\bibitem[OpenAI(2023)]%
        {GPT4}
\bibfield{author}{\bibinfo{person}{OpenAI}.} \bibinfo{year}{2023}\natexlab{}.
\newblock \showarticletitle{{GPT-4} Technical Report}.
\newblock \bibinfo{journal}{\emph{CoRR}} (\bibinfo{year}{2023}).
\newblock


\bibitem[Perron et~al\mbox{.}(2023)]%
        {CP-SAT}
\bibfield{author}{\bibinfo{person}{Laurent Perron}, \bibinfo{person}{Fr{\'e}d{\'e}ric Didier}, {and} \bibinfo{person}{Steven Gay}.} \bibinfo{year}{2023}\natexlab{}.
\newblock \showarticletitle{The CP-SAT-LP solver (invited talk)}. In \bibinfo{booktitle}{\emph{29th International Conference on Principles and Practice of Constraint Programming (CP 2023)}}. Schloss Dagstuhl--Leibniz-Zentrum f{\"u}r Informatik, \bibinfo{pages}{3--1}.
\newblock


\bibitem[Rajput et~al\mbox{.}(2023)]%
        {TIGER}
\bibfield{author}{\bibinfo{person}{Shashank Rajput}, \bibinfo{person}{Nikhil Mehta}, \bibinfo{person}{Anima Singh}, \bibinfo{person}{Raghunandan~Hulikal Keshavan}, \bibinfo{person}{Trung Vu}, \bibinfo{person}{Lukasz Heldt}, \bibinfo{person}{Lichan Hong}, \bibinfo{person}{Yi Tay}, \bibinfo{person}{Vinh~Q. Tran}, \bibinfo{person}{Jonah Samost}, \bibinfo{person}{Maciej Kula}, \bibinfo{person}{Ed~H. Chi}, {and} \bibinfo{person}{Mahesh Sathiamoorthy}.} \bibinfo{year}{2023}\natexlab{}.
\newblock \showarticletitle{Recommender Systems with Generative Retrieval}. In \bibinfo{booktitle}{\emph{NeurIPS}}, \bibfield{editor}{\bibinfo{person}{Alice Oh}, \bibinfo{person}{Tristan Naumann}, \bibinfo{person}{Amir Globerson}, \bibinfo{person}{Kate Saenko}, \bibinfo{person}{Moritz Hardt}, {and} \bibinfo{person}{Sergey Levine}} (Eds.).
\newblock


\bibitem[Rendle et~al\mbox{.}(2012)]%
        {BPR}
\bibfield{author}{\bibinfo{person}{Steffen Rendle}, \bibinfo{person}{Christoph Freudenthaler}, \bibinfo{person}{Zeno Gantner}, {and} \bibinfo{person}{Lars Schmidt{-}Thieme}.} \bibinfo{year}{2012}\natexlab{}.
\newblock \showarticletitle{{BPR:} Bayesian Personalized Ranking from Implicit Feedback}.
\newblock \bibinfo{journal}{\emph{CoRR}}  \bibinfo{volume}{abs/1205.2618} (\bibinfo{year}{2012}).
\newblock


\bibitem[Shao et~al\mbox{.}(2024)]%
        {deepseekmath}
\bibfield{author}{\bibinfo{person}{Zhihong Shao}, \bibinfo{person}{Peiyi Wang}, \bibinfo{person}{Qihao Zhu}, \bibinfo{person}{Runxin Xu}, \bibinfo{person}{Junxiao Song}, \bibinfo{person}{Xiao Bi}, \bibinfo{person}{Haowei Zhang}, \bibinfo{person}{Mingchuan Zhang}, \bibinfo{person}{YK Li}, \bibinfo{person}{Yang Wu}, {et~al\mbox{.}}} \bibinfo{year}{2024}\natexlab{}.
\newblock \showarticletitle{Deepseekmath: Pushing the limits of mathematical reasoning in open language models}.
\newblock \bibinfo{journal}{\emph{arXiv preprint arXiv:2402.03300}} (\bibinfo{year}{2024}).
\newblock


\bibitem[Su et~al\mbox{.}(2026)]%
        {vectorizing_beam_search}
\bibfield{author}{\bibinfo{person}{Zhengyang Su}, \bibinfo{person}{Isay Katsman}, \bibinfo{person}{Yueqi Wang}, \bibinfo{person}{Ruining He}, \bibinfo{person}{Lukasz Heldt}, \bibinfo{person}{Raghunandan Keshavan}, \bibinfo{person}{Shao-Chuan Wang}, \bibinfo{person}{Xinyang Yi}, \bibinfo{person}{Mingyan Gao}, \bibinfo{person}{Onkar Dalal}, {et~al\mbox{.}}} \bibinfo{year}{2026}\natexlab{}.
\newblock \showarticletitle{Vectorizing the Trie: Efficient Constrained Decoding for LLM-based Generative Retrieval on Accelerators}.
\newblock \bibinfo{journal}{\emph{arXiv preprint arXiv:2602.22647}} (\bibinfo{year}{2026}).
\newblock


\bibitem[Tan et~al\mbox{.}(2025)]%
        {ReRe}
\bibfield{author}{\bibinfo{person}{Junfei Tan}, \bibinfo{person}{Yuxin Chen}, \bibinfo{person}{An Zhang}, \bibinfo{person}{Junguang Jiang}, \bibinfo{person}{Bin Liu}, \bibinfo{person}{Ziru Xu}, \bibinfo{person}{Han Zhu}, \bibinfo{person}{Jian Xu}, \bibinfo{person}{Bo Zheng}, {and} \bibinfo{person}{Xiang Wang}.} \bibinfo{year}{2025}\natexlab{}.
\newblock \showarticletitle{Reinforced preference optimization for recommendation}.
\newblock \bibinfo{journal}{\emph{arXiv preprint arXiv:2510.12211}} (\bibinfo{year}{2025}).
\newblock


\bibitem[Tan et~al\mbox{.}(2024)]%
        {IDGenRec}
\bibfield{author}{\bibinfo{person}{Juntao Tan}, \bibinfo{person}{Shuyuan Xu}, \bibinfo{person}{Wenyue Hua}, \bibinfo{person}{Yingqiang Ge}, \bibinfo{person}{Zelong Li}, {and} \bibinfo{person}{Yongfeng Zhang}.} \bibinfo{year}{2024}\natexlab{}.
\newblock \showarticletitle{Idgenrec: Llm-recsys alignment with textual id learning}. In \bibinfo{booktitle}{\emph{Proceedings of the 47th International ACM SIGIR Conference on Research and Development in Information Retrieval}}. \bibinfo{pages}{355--364}.
\newblock


\bibitem[Xie et~al\mbox{.}(2026)]%
        {agentictagger}
\bibfield{author}{\bibinfo{person}{Zhouhang Xie}, \bibinfo{person}{Bo Peng}, \bibinfo{person}{Zhankui He}, \bibinfo{person}{Ziqi Chen}, \bibinfo{person}{Alice Han}, \bibinfo{person}{Isabella Ye}, \bibinfo{person}{Benjamin Coleman}, \bibinfo{person}{Noveen Sachdeva}, \bibinfo{person}{Fernando Pereira}, \bibinfo{person}{Julian McAuley}, {et~al\mbox{.}}} \bibinfo{year}{2026}\natexlab{}.
\newblock \showarticletitle{AgenticTagger: Structured Item Representation for Recommendation with LLM Agents}.
\newblock \bibinfo{journal}{\emph{arXiv preprint arXiv:2602.05945}} (\bibinfo{year}{2026}).
\newblock


\bibitem[Xu et~al\mbox{.}(2025)]%
        {iagent}
\bibfield{author}{\bibinfo{person}{Wujiang Xu}, \bibinfo{person}{Yunxiao Shi}, \bibinfo{person}{Zujie Liang}, \bibinfo{person}{Xuying Ning}, \bibinfo{person}{Kai Mei}, \bibinfo{person}{Kun Wang}, \bibinfo{person}{Xi Zhu}, \bibinfo{person}{Min Xu}, {and} \bibinfo{person}{Yongfeng Zhang}.} \bibinfo{year}{2025}\natexlab{}.
\newblock \showarticletitle{iagent: Llm agent as a shield between user and recommender systems}. In \bibinfo{booktitle}{\emph{Findings of the Association for Computational Linguistics: ACL 2025}}. \bibinfo{pages}{18056--18084}.
\newblock


\bibitem[Yang et~al\mbox{.}(2025)]%
        {qwen3}
\bibfield{author}{\bibinfo{person}{An Yang}, \bibinfo{person}{Anfeng Li}, \bibinfo{person}{Baosong Yang}, \bibinfo{person}{Beichen Zhang}, \bibinfo{person}{Binyuan Hui}, \bibinfo{person}{Bo Zheng}, \bibinfo{person}{Bowen Yu}, \bibinfo{person}{Chang Gao}, \bibinfo{person}{Chengen Huang}, \bibinfo{person}{Chenxu Lv}, {et~al\mbox{.}}} \bibinfo{year}{2025}\natexlab{}.
\newblock \showarticletitle{Qwen3 technical report}.
\newblock \bibinfo{journal}{\emph{arXiv preprint arXiv:2505.09388}} (\bibinfo{year}{2025}).
\newblock


\bibitem[Yang et~al\mbox{.}(2020)]%
        {yang2020mixed}
\bibfield{author}{\bibinfo{person}{Ji Yang}, \bibinfo{person}{Xinyang Yi}, \bibinfo{person}{Derek Zhiyuan~Cheng}, \bibinfo{person}{Lichan Hong}, \bibinfo{person}{Yang Li}, \bibinfo{person}{Simon Xiaoming~Wang}, \bibinfo{person}{Taibai Xu}, {and} \bibinfo{person}{Ed~H Chi}.} \bibinfo{year}{2020}\natexlab{}.
\newblock \showarticletitle{Mixed negative sampling for learning two-tower neural networks in recommendations}. In \bibinfo{booktitle}{\emph{Companion proceedings of the web conference 2020}}. \bibinfo{pages}{441--447}.
\newblock


\bibitem[Yi et~al\mbox{.}(2019)]%
        {yi2019sampling}
\bibfield{author}{\bibinfo{person}{Xinyang Yi}, \bibinfo{person}{Ji Yang}, \bibinfo{person}{Lichan Hong}, \bibinfo{person}{Derek~Zhiyuan Cheng}, \bibinfo{person}{Lukasz Heldt}, \bibinfo{person}{Aditee Kumthekar}, \bibinfo{person}{Zhe Zhao}, \bibinfo{person}{Li Wei}, {and} \bibinfo{person}{Ed Chi}.} \bibinfo{year}{2019}\natexlab{}.
\newblock \showarticletitle{Sampling-bias-corrected neural modeling for large corpus item recommendations}. In \bibinfo{booktitle}{\emph{Proceedings of the 13th ACM conference on recommender systems}}. \bibinfo{pages}{269--277}.
\newblock


\bibitem[Yue et~al\mbox{.}(2023)]%
        {LlamaRec}
\bibfield{author}{\bibinfo{person}{Zhenrui Yue}, \bibinfo{person}{Sara Rabhi}, \bibinfo{person}{Gabriel de Souza~Pereira Moreira}, \bibinfo{person}{Dong Wang}, {and} \bibinfo{person}{Even Oldridge}.} \bibinfo{year}{2023}\natexlab{}.
\newblock \showarticletitle{Llamarec: Two-stage recommendation using large language models for ranking}.
\newblock \bibinfo{journal}{\emph{arXiv preprint arXiv:2311.02089}} (\bibinfo{year}{2023}).
\newblock


\bibitem[Zhai et~al\mbox{.}(2024)]%
        {HSTU}
\bibfield{author}{\bibinfo{person}{Jiaqi Zhai}, \bibinfo{person}{Lucy Liao}, \bibinfo{person}{Xing Liu}, \bibinfo{person}{Yueming Wang}, \bibinfo{person}{Rui Li}, \bibinfo{person}{Xuan Cao}, \bibinfo{person}{Leon Gao}, \bibinfo{person}{Zhaojie Gong}, \bibinfo{person}{Fangda Gu}, \bibinfo{person}{Michael He}, {et~al\mbox{.}}} \bibinfo{year}{2024}\natexlab{}.
\newblock \showarticletitle{Actions speak louder than words: Trillion-parameter sequential transducers for generative recommendations}.
\newblock \bibinfo{journal}{\emph{arXiv preprint arXiv:2402.17152}} (\bibinfo{year}{2024}).
\newblock


\bibitem[Zhang et~al\mbox{.}(2023b)]%
        {agent4rec}
\bibfield{author}{\bibinfo{person}{An Zhang}, \bibinfo{person}{Leheng Sheng}, \bibinfo{person}{Yuxin Chen}, \bibinfo{person}{Hao Li}, \bibinfo{person}{Yang Deng}, \bibinfo{person}{Xiang Wang}, {and} \bibinfo{person}{Tat{-}Seng Chua}.} \bibinfo{year}{2023}\natexlab{b}.
\newblock \showarticletitle{On Generative Agents in Recommendation}.
\newblock \bibinfo{journal}{\emph{CoRR}}  \bibinfo{volume}{abs/2310.10108} (\bibinfo{year}{2023}).
\newblock


\bibitem[Zhang et~al\mbox{.}(2023a)]%
        {agentcf}
\bibfield{author}{\bibinfo{person}{Junjie Zhang}, \bibinfo{person}{Yupeng Hou}, \bibinfo{person}{Ruobing Xie}, \bibinfo{person}{Wenqi Sun}, \bibinfo{person}{Julian~J. McAuley}, \bibinfo{person}{Wayne~Xin Zhao}, \bibinfo{person}{Leyu Lin}, {and} \bibinfo{person}{Ji{-}Rong Wen}.} \bibinfo{year}{2023}\natexlab{a}.
\newblock \showarticletitle{AgentCF: Collaborative Learning with Autonomous Language Agents for Recommender Systems}.
\newblock \bibinfo{journal}{\emph{CoRR}}  \bibinfo{volume}{abs/2310.09233} (\bibinfo{year}{2023}).
\newblock


\bibitem[Zhang et~al\mbox{.}(2023c)]%
        {InstructRec}
\bibfield{author}{\bibinfo{person}{Junjie Zhang}, \bibinfo{person}{Ruobing Xie}, \bibinfo{person}{Yupeng Hou}, \bibinfo{person}{Wayne~Xin Zhao}, \bibinfo{person}{Leyu Lin}, {and} \bibinfo{person}{Ji{-}Rong Wen}.} \bibinfo{year}{2023}\natexlab{c}.
\newblock \showarticletitle{Recommendation as Instruction Following: {A} Large Language Model Empowered Recommendation Approach}.
\newblock \bibinfo{journal}{\emph{CoRR}}  \bibinfo{volume}{abs/2305.07001} (\bibinfo{year}{2023}).
\newblock


\bibitem[Zhang et~al\mbox{.}(2026a)]%
        {why-think-hurts}
\bibfield{author}{\bibinfo{person}{Luankang Zhang}, \bibinfo{person}{Yonghao Huang}, \bibinfo{person}{Hang Lv}, \bibinfo{person}{Mingjia Yin}, \bibinfo{person}{Liangyue Li}, \bibinfo{person}{Zulong Chen}, \bibinfo{person}{Hao Wang}, {and} \bibinfo{person}{Enhong Chen}.} \bibinfo{year}{2026}\natexlab{a}.
\newblock \showarticletitle{Why Thinking Hurts? Diagnosing and Rectifying the Reasoning Shift in Foundation Recommender Models}.
\newblock \bibinfo{journal}{\emph{arXiv preprint arXiv:2602.16587}} (\bibinfo{year}{2026}).
\newblock


\bibitem[Zhang et~al\mbox{.}({[n.\,d.]})]%
        {rec_ID_survey}
\bibfield{author}{\bibinfo{person}{Taiyan Zhang}, \bibinfo{person}{Hongtao Wang}, \bibinfo{person}{Yunqian Fan}, \bibinfo{person}{Kunda Yang}, \bibinfo{person}{Jichuan Zeng}, {and} \bibinfo{person}{Renchi Yang}.} \bibinfo{year}{[n.\,d.]}\natexlab{}.
\newblock \showarticletitle{A Survey of Item Identifiers in Generative Recommendation: Construction, Alignment, and Generation}.
\newblock  (\bibinfo{year}{[n.\,d.]}).
\newblock


\bibitem[Zhang et~al\mbox{.}(2026b)]%
        {TID}
\bibfield{author}{\bibinfo{person}{Zhiyang Zhang}, \bibinfo{person}{Junda She}, \bibinfo{person}{Kuo Cai}, \bibinfo{person}{Bo Chen}, \bibinfo{person}{Shiyao Wang}, \bibinfo{person}{Xinchen Luo}, \bibinfo{person}{Qiang Luo}, \bibinfo{person}{Ruiming Tang}, \bibinfo{person}{Han Li}, \bibinfo{person}{Kun Gai}, {et~al\mbox{.}}} \bibinfo{year}{2026}\natexlab{b}.
\newblock \showarticletitle{Unleashing the Native Recommendation Potential: LLM-Based Generative Recommendation via Structured Term Identifiers}.
\newblock \bibinfo{journal}{\emph{arXiv preprint arXiv:2601.06798}} (\bibinfo{year}{2026}).
\newblock


\bibitem[Zheng et~al\mbox{.}(2023)]%
        {LC-Rec}
\bibfield{author}{\bibinfo{person}{Bowen Zheng}, \bibinfo{person}{Yupeng Hou}, \bibinfo{person}{Hongyu Lu}, \bibinfo{person}{Yu Chen}, \bibinfo{person}{Wayne~Xin Zhao}, \bibinfo{person}{Ming Chen}, {and} \bibinfo{person}{Ji{-}Rong Wen}.} \bibinfo{year}{2023}\natexlab{}.
\newblock \showarticletitle{Adapting Large Language Models by Integrating Collaborative Semantics for Recommendation}.
\newblock \bibinfo{journal}{\emph{CoRR}}  \bibinfo{volume}{abs/2311.09049} (\bibinfo{year}{2023}).
\newblock


\bibitem[Zhou et~al\mbox{.}(2025a)]%
        {openonerec}
\bibfield{author}{\bibinfo{person}{Guorui Zhou}, \bibinfo{person}{Honghui Bao}, \bibinfo{person}{Jiaming Huang}, \bibinfo{person}{Jiaxin Deng}, \bibinfo{person}{Jinghao Zhang}, \bibinfo{person}{Junda She}, \bibinfo{person}{Kuo Cai}, \bibinfo{person}{Lejian Ren}, \bibinfo{person}{Lu Ren}, \bibinfo{person}{Qiang Luo}, {et~al\mbox{.}}} \bibinfo{year}{2025}\natexlab{a}.
\newblock \showarticletitle{OpenOneRec Technical Report}.
\newblock \bibinfo{journal}{\emph{arXiv preprint arXiv:2512.24762}} (\bibinfo{year}{2025}).
\newblock


\bibitem[Zhou et~al\mbox{.}(2025b)]%
        {OneRec-v2}
\bibfield{author}{\bibinfo{person}{Guorui Zhou}, \bibinfo{person}{Hengrui Hu}, \bibinfo{person}{Hongtao Cheng}, \bibinfo{person}{Huanjie Wang}, \bibinfo{person}{Jiaxin Deng}, \bibinfo{person}{Jinghao Zhang}, \bibinfo{person}{Kuo Cai}, \bibinfo{person}{Lejian Ren}, \bibinfo{person}{Lu Ren}, \bibinfo{person}{Liao Yu}, {et~al\mbox{.}}} \bibinfo{year}{2025}\natexlab{b}.
\newblock \showarticletitle{Onerec-v2 technical report}.
\newblock \bibinfo{journal}{\emph{arXiv preprint arXiv:2508.20900}} (\bibinfo{year}{2025}).
\newblock


\end{thebibliography}

\appendix

\section{LLM Prompt Samples}
\label{app:prompt}

\begin{prefbox}[LLM Prompts for Feature Filtering]
You are an expert product keywords extractor. Your task is to extract informative keywords from this product's metadata (only from the item title and description fields) and dump useless words. Please follow ALL guidelines carefully:

GUIDELINES:
1) WORD FORM: All words must be in their base form.
2) Only keep words that are informative to the product itself and useful for recommendation task. Dump generic or vague terms.
3) CONTENT FOCUS: words related to product category, function, features and selling point can be considered informative.
4) NO ADDITIONAL TEXT: Do not include any explanations, thoughts, or other content.

\end{prefbox}

Let $\langle r_{j}\rangle$ denote the trie path used to represent item $i_j$. A sample prompt is of the following form:
\begin{prefbox}[LLM Prompts for Sequential Recommendation]
\textbf{Input:} The user has interacted with the following items in chronological order:
$\langle r_{1}\rangle,\langle r_{2}\rangle,\ldots,\langle r_{t}\rangle$.
Please predict the next item the user may prefer.\\
\textbf{Response:} $\langle r_{t+1}\rangle$
\end{prefbox}

\section{Complete Theoretical Analysis}\label{app:proof}

In this appendix, we provide the complete forms and proofs of
Theorems.
Throughout, we strictly follow the notation used in the main text.

\subsection{Additional Definitions and Regularity Conditions}

Let $I$ be the item set, and let $Y\in I$ be the target item with prior
\[
\Pr(Y=i)=p_i,\qquad p_i\ge 0,\qquad \sum_{i\in I}p_i=1.
\]
For any nonempty subset $U\subseteq I$, define
\[
p(U):=\sum_{i\in U}p_i.
\]

A feasible split of $U$ is a partition
\[
\mathcal P=\{C_1,\dots,C_k\}\in \Pi(U),
\]
where $C_j\neq \emptyset$, $C_j\subseteq U$, $C_a\cap C_b=\emptyset$ for $a\neq b$, and
\[
\bigcup_{j=1}^k C_j=U.
\]

Fix beam width $B\ge 1$, and define the one-step beam survival factor
\[
\phi_B(k)=\min\left(1,\frac{B}{k}\right).
\]

For any feasible tree $T$, let $d(i)$ denote the depth of leaf $i$, and let
$k_t(i)$ denote the branching factor encountered by item $i$ at depth $t$.
Its beam-aware correct rate is
\[
R(T)=\sum_{i\in I} p_i \prod_{t=1}^{d(i)} \phi_B\!\bigl(k_t(i)\bigr).
\]

For any nonempty subset $U\subseteq I$, the optimal continuation value is
\[
V(U):=\max_{T\text{ feasible on }U}\Pr(\mathrm{recover}\ Y\mid Y\in U,\;T).
\]

For any feasible split $\mathcal P=\{C_1,\dots,C_k\}\in\Pi(U)$, define
\[
\operatorname{Val}(U,\mathcal P)
:=
\phi_B(k)\sum_{j=1}^k \frac{p(C_j)}{p(U)}V(C_j).
\]

For \cref{thm:local-approx}, we use the following regularity conditions.

\begin{assumption}[Beam-limited regime]
\label{ass:beam-limited}
For every non-singleton subset $U\subseteq I$ and every feasible split
$\mathcal P\in\Pi(U)$,
\[
|\mathcal P|\ge B.
\]
Hence
\[
\phi_B(|\mathcal P|)=\frac{B}{|\mathcal P|}.
\]
\end{assumption}

\begin{assumption}[Bounded continuation heterogeneity]
\label{ass:heterogeneity}
There exists a constant $\rho\ge 1$ such that for every non-singleton subset
$U\subseteq I$ and every two feasible splits $\mathcal P,\mathcal Q\in\Pi(U)$,
\[
W(U,\mathcal P)\le \rho\, W(U,\mathcal Q),
\]
where
\[
W(U,\mathcal P):=
\sum_{C\in\mathcal P}\frac{p(C)}{p(U)}V(C).
\]
\end{assumption}

For \cref{thm:variable-depth}, we use the following condition.

\begin{assumption}[Nontrivial extension]
\label{ass:nontrivial-extension}
Whenever a shallower leaf is extended only to match the depth of a deeper leaf,
at least one added level has branching factor $k>B$. Equivalently, there exists
$\eta<1$ such that every such forced extension multiplies the corresponding leaf
success probability by at most $\eta$.
\end{assumption}

\subsection{Proof of Theorem 1}

We first restate the theorem in complete form.

\begin{customthm}{1}[Optimal tree construction]
\label{thm:optimal-tree-app}
For any nonempty subset $U\subseteq I$, $V(U)$ denotes the maximum achievable beam search success rate conditioned on $Y\in U$. For any feasible split $\mathcal P\in\Pi(U)$, $\operatorname{Val}(U,\mathcal P)$ denotes the success rate obtained by first applying split $\mathcal P$ at $U$ and then continuing optimally in each resulting child subset. For every non-singleton subset $U\subseteq I$,
\[
V(U)
=
\max_{\mathcal P\in \Pi(U)}
\operatorname{Val}(U,\mathcal P).
\]
Choosing at each subset $U$ a maximizing split in $\Pi(U)$ with Bellman dynamic programming yields a globally optimal tree, and the optimal correct rate on the full item set is $V(I)$.
\end{customthm}

\begin{proof}
We prove the Bellman recursion by induction on $|U|$.

\paragraph{Base case.}
If $U=\{i\}$ is a singleton, then the target item is already uniquely identified.
No further split is needed, so
\[
V(\{i\})=1.
\]

\paragraph{Inductive hypothesis.}
Assume that for every nonempty subset $W\subseteq I$ with $|W|<|U|$, the value
$V(W)$ is given by the Bellman recursion, and there exists a feasible subtree on
$W$ attaining $V(W)$.

\paragraph{Inductive step.}
Now consider any non-singleton subset $U\subseteq I$.
Take an arbitrary feasible tree on $U$. Its root must choose some feasible split
\[
\mathcal P=\{C_1,\dots,C_k\}\in\Pi(U).
\]
Conditioned on $Y\in U$, the beam keeps the correct child with probability
$\phi_B(k)$. Given survival at this level, the target lies in child $C_j$ with
conditional probability
\[
\Pr(Y\in C_j\mid Y\in U)=\frac{p(C_j)}{p(U)}.
\]
After entering $C_j$, the best achievable continuation value is $V(C_j)$ by the
inductive hypothesis. Therefore, the success rate of any tree whose first split is
$\mathcal P$ is at most
\[
\phi_B(k)\sum_{j=1}^k \frac{p(C_j)}{p(U)}V(C_j)
=
\operatorname{Val}(U,\mathcal P).
\]
Since this upper bound holds for every feasible first split $\mathcal P\in\Pi(U)$,
we obtain
\[
V(U)\le \max_{\mathcal P\in\Pi(U)} \operatorname{Val}(U,\mathcal P).
\]

To prove the reverse inequality, let
\[
\mathcal P^*(U)\in \arg\max_{\mathcal P\in\Pi(U)} \operatorname{Val}(U,\mathcal P).
\]
By the inductive hypothesis, for each child $C_j\in \mathcal P^*(U)$ there exists
a feasible subtree attaining $V(C_j)$. Construct a tree on $U$ by using
$\mathcal P^*(U)$ at the root and these optimal subtrees below each child.
The resulting success rate is exactly
\[
\operatorname{Val}(U,\mathcal P^*(U))
=
\max_{\mathcal P\in\Pi(U)} \operatorname{Val}(U,\mathcal P).
\]
Hence
\[
V(U)\ge \max_{\mathcal P\in\Pi(U)} \operatorname{Val}(U,\mathcal P).
\]

Combining the two inequalities yields
\[
V(U)
=
\max_{\mathcal P\in \Pi(U)}
\operatorname{Val}(U,\mathcal P).
\]

Thus, for every subset $U$, an optimal solution is obtained by choosing a maximizing
split and then solving each child subset optimally. Applying this recursion to the
full item set $I$ gives the Bellman dynamic programming solution, and the optimal
correct rate is $V(I)$.
\end{proof}

\subsection{Proof of Theorem 2}

We next restate the local approximation theorem in complete form.

\begin{customthm}{2}[Local approximation of greedy cover set]
\label{thm:local-approx-app}
Under Assumptions~\ref{ass:beam-limited}--\ref{ass:heterogeneity}, for every
non-singleton subset $U\subseteq I$, let
\[
\mathcal P^*(U)\in \arg\max_{\mathcal P\in\Pi(U)} \operatorname{Val}(U,\mathcal P)
\]
be a locally optimal split. Then the greedy set-cover split $\mathcal P^g(U)$ satisfies
\[
\operatorname{Val}(U,\mathcal P^g(U))
\ge
\frac{1}{\rho \alpha_U}
\operatorname{Val}(U,\mathcal P^*(U)),
\]
where $\rho\ge 1$ is the constant in Assumption~\ref{ass:heterogeneity}, and
$\alpha_U\ge 1$ satisfies
\[
|\mathcal P^g(U)|\le \alpha_U \min_{\mathcal P\in\Pi(U)} |\mathcal P|.
\]
Thus greedy cover set is a provable local approximation when the true prior
$\{p_i\}$ is unknown.
\end{customthm}

\begin{proof}
Fix a non-singleton subset $U\subseteq I$.

By Assumption~\ref{ass:beam-limited}, every feasible split $\mathcal P\in\Pi(U)$
satisfies $|\mathcal P|\ge B$, and therefore
\[
\phi_B(|\mathcal P|)=\frac{B}{|\mathcal P|}.
\]
Hence, for any feasible split $\mathcal P$,
\[
\operatorname{Val}(U,\mathcal P)
=
\frac{B}{|\mathcal P|}\,W(U,\mathcal P),
\]
where
\[
W(U,\mathcal P)=\sum_{C\in\mathcal P}\frac{p(C)}{p(U)}V(C).
\]

Applying this to the greedy split $\mathcal P^g(U)$ gives
\[
\operatorname{Val}(U,\mathcal P^g(U))
=
\frac{B}{|\mathcal P^g(U)|}\,W(U,\mathcal P^g(U)).
\]

By Assumption~\ref{ass:heterogeneity},
\[
W(U,\mathcal P^g(U))
\ge \frac{1}{\rho}W(U,\mathcal P^*(U)).
\]
Therefore,
\[
\operatorname{Val}(U,\mathcal P^g(U))
\ge
\frac{B}{|\mathcal P^g(U)|}\cdot \frac{1}{\rho}W(U,\mathcal P^*(U)).
\]

Now use the definition of $\alpha_U$:
\[
|\mathcal P^g(U)|
\le
\alpha_U \min_{\mathcal P\in\Pi(U)} |\mathcal P|.
\]
Since $\mathcal P^*(U)$ is itself a feasible split,
\[
\min_{\mathcal P\in\Pi(U)} |\mathcal P|
\le
|\mathcal P^*(U)|.
\]
Hence
\[
|\mathcal P^g(U)|\le \alpha_U |\mathcal P^*(U)|.
\]
Substituting this into the previous inequality yields
\[
\operatorname{Val}(U,\mathcal P^g(U))
\ge
\frac{B}{\alpha_U |\mathcal P^*(U)|}\cdot \frac{1}{\rho}W(U,\mathcal P^*(U)).
\]

On the other hand, again by Assumption~\ref{ass:beam-limited},
\[
\operatorname{Val}(U,\mathcal P^*(U))
=
\frac{B}{|\mathcal P^*(U)|}W(U,\mathcal P^*(U)).
\]
Combining the two displays gives
\[
\operatorname{Val}(U,\mathcal P^g(U))
\ge
\frac{1}{\rho\alpha_U}
\operatorname{Val}(U,\mathcal P^*(U)).
\]

This proves that the greedy cover-set split attains at least a
$1/(\rho\alpha_U)$ fraction of the optimal local beam-aware value.
\end{proof}

\subsection{Proof of Theorem 3}

We finally restate the variable-depth theorem in complete form.

\begin{customthm}{3}[Variable depth is preferable to fixed depth]
\label{thm:variable-depth-app}
Let $\mathcal T$ be the set of all feasible tries, and let
$\mathcal T_{\mathrm{eq}}\subseteq \mathcal T$ be the subset of feasible trees
whose leaves have the same depth. Then
\[
\max_{T\in\mathcal T_{\mathrm{eq}}}R(T)\le \max_{T\in\mathcal T}R(T).
\]
Moreover, under Assumption~\ref{ass:nontrivial-extension}, if a feasible tree $T$
has leaves of different depths, then any equal-depth tree obtained by extending the
shallower leaves of $T$ satisfies
\[
R(\widetilde T)<R(T).
\]
Therefore, both the optimal construction and the greedy approximation naturally admit
variable leaf depths, while forcing fixed depth will be suboptimal.
\end{customthm}

\begin{proof}
The first inequality is immediate from set inclusion:
\[
\mathcal T_{\mathrm{eq}}\subseteq \mathcal T.
\]
Therefore, maximizing the same objective $R(T)$ over the restricted family
$\mathcal T_{\mathrm{eq}}$ cannot exceed maximizing it over the full family
$\mathcal T$, i.e.,
\[
\max_{T\in\mathcal T_{\mathrm{eq}}}R(T)\le \max_{T\in\mathcal T}R(T).
\]

We now prove the strict inequality for a fixed variable-depth tree $T$.
Assume that $T$ has leaves of different depths. Let $L_{\max}$ be the maximum leaf
depth in $T$, and let $\widetilde T$ be any equal-depth tree obtained by extending
the shallower leaves of $T$ so that all leaves have depth $L_{\max}$.

Consider any item $i$ whose leaf depth in $T$ is strictly smaller than $L_{\max}$.
By construction, the path of $i$ in $\widetilde T$ must contain at least one added
level. By Assumption~\ref{ass:nontrivial-extension}, at least one of the added
levels has branching factor $k>B$, and therefore contributes a multiplicative factor
\[
\phi_B(k)=\frac{B}{k}\le \eta<1.
\]
Hence the success probability of item $i$ in $\widetilde T$ is strictly smaller than
that in $T$:
\[
\prod_{t=1}^{d_{\widetilde T}(i)} \phi_B\!\bigl(k_t^{\widetilde T}(i)\bigr)
<
\prod_{t=1}^{d_T(i)} \phi_B\!\bigl(k_t^{T}(i)\bigr).
\]

For any item $j$ whose original leaf depth in $T$ is already $L_{\max}$, the branch
is unchanged, so its success probability remains the same:
\[
\prod_{t=1}^{d_{\widetilde T}(j)} \phi_B\!\bigl(k_t^{\widetilde T}(j)\bigr)
=
\prod_{t=1}^{d_T(j)} \phi_B\!\bigl(k_t^{T}(j)\bigr).
\]

Therefore, term-by-term in the definition of $R(T)$, at least one positive-probability
item suffers a strict decrease while no item improves. Summing over all items yields
\[
R(\widetilde T)<R(T).
\]

This proves that once a variable-depth tree is given, forcing equal depth by extending
shallower leaves can only decrease the beam-aware success rate under the stated
nontrivial-extension condition.
\end{proof}

\section{Dataset Details}
\label{app:dataset}
\begin{table}[h]
\centering
\caption{Statistics of the datasets used in our experiments.}
\begin{tabular}{lrrr}
\toprule
\textbf{Dataset} & \textbf{\# users} & \textbf{\# items} & \textbf{\# actions}  \\
\midrule
\textit{Beauty} & 22,363 & 12,101 & 198,502 \\
\textit{Toys and Games} & 19,412 & 11,924 & 167,597 \\
\textit{Sports and Outdoors} & 35,598 & 18,357 & 296,337\\\bottomrule
\end{tabular}
\label{tab:dataset_stats}
\end{table}

\begin{table}[h]
\centering
\caption{Comparisons between the number of items and trie leaves~(item IDs). Our method does not significantly increase the number of item representations.}
\label{tab:leaves num}
\begin{tabular}{c|ccc}
\toprule
 & \textbf{Beauty} &\textbf{Sports}  & \textbf{Toys}  \\
\midrule
\#Items & 12,101 &18,357  &11,923  \\
\midrule
\#Trie Leaves &13,823  & 31,306 &15,824  \\
\bottomrule
\end{tabular}

\end{table}

\begin{table}[t]
\centering
\caption{ID Collision rate comparison across different textual ID formulations.}
\label{tab:col_ratio}
\small
\setlength{\tabcolsep}{6pt}
\renewcommand{\arraystretch}{1.2}
\begin{tabular}{lccc}
\toprule
\textbf{Dataset} & \textbf{Ours } & \textbf{GRLM} & \textbf{Title-10 token } \\
\midrule
\textbf{Beauty} & 3.38\% & 3.64\% & 1.63\% \\
\textbf{Sports} & 3.20\% & 4.30\% & 2.22\% \\
\textbf{Toys}   & 3.81\% & 2.29\% & 2.00\% \\
\bottomrule
\end{tabular}
\end{table}

\section{Performance in Cold-Start Settings~(RQ5)}\label{app:RQ5}

\begin{table}[t]
\caption{Experiments on the cold-start items. BONSAI shows strong capability to handle new items.}

\centering
\begin{tabular}{c|ccc}
\hline
       & OneRec-Think   & GRLM & BONSAI      \\ \hline
Beauty  & 0.028 & 0.032 &\textbf{0.038} \\ 
Sports  & 0.016 & 0.018  &\textbf{0.022} \\ 
Toys  & 0.028 & 0.034  & \textbf{0.038}\\ \hline
\end{tabular}
\label{tab:cold_start}
\end{table}

We next evaluate the cold-start capability of BONSAI. Based on the three datasets used in \cref{subsec:overall_perf}, we construct a cold-start test set by randomly selecting 500 items from each dataset, removing them from all training sequences, and evaluating model performance on these held-out items. We then train all models and report Recall@10 on the cold-start items.
As shown in Table~\ref{tab:cold_start}, BONSAI consistently achieves the best performance, outperforming OneRec-Think and GRLM, the two strongest baselines in this setting. The results suggest that BONSAI effectively leverages semantic understanding and its adaptive trie structure to mitigate the cold-start problem, leading to more accurate recommendations for newly introduced items.

\section{Scaling with Model and Data Size~(RQ6)}\label{app:RQ6}

\begin{figure}[t]
\begin{center}
\includegraphics[width=0.45\textwidth]{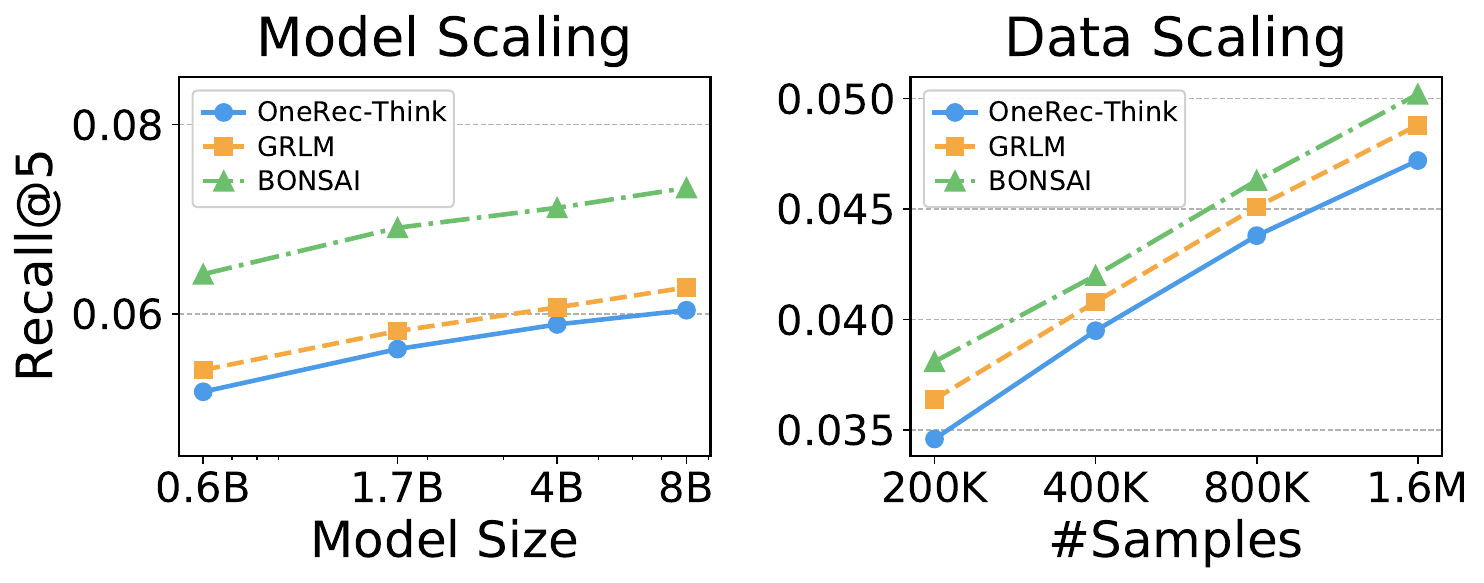}
\caption{Experiments on model and data scaling. BONSAI shows strong scaling behaviors consistently.}
\label{fig:scaling}
\end{center}
\end{figure}

We further examine the scaling behavior of BONSAI relative to state-of-the-art baselines. Specifically, for model scaling, we increase the backbone size from 0.6B to 8B using Qwen3 models~\citep{qwen3}. For data scaling, we increase the amount of training data using Amazon-M2~\citep{AmazonM2}, an industry-scale recommendation dataset. The results are shown in Figure~\ref{fig:scaling}.
BONSAI improves consistently under both model and data scaling, demonstrating a strong ability to benefit from larger backbones and richer training data. Notably, it maintains the best Recall@5 across all scaling settings, indicating that its advantage is robust rather than confined to a specific model size or data regime. These results suggest that BONSAI is highly compatible with scaling and can effectively translate additional model capacity and training data into steady performance gains.

\end{document}